\documentclass[10pt,english,twocolumn,twoside]{IEEEtran}
\usepackage[latin9]{inputenc}
\usepackage{float}
\usepackage{amsmath}
\usepackage{amssymb}
\usepackage{graphicx}
\usepackage{color}
\usepackage{url}
\usepackage[nospace,compress]{cite}

\makeatletter

\floatstyle{ruled}
\newfloat{algorithm}{tbp}{loa}
\providecommand{\algorithmname}{Algorithm}
\floatname{algorithm}{\protect\algorithmname}
\usepackage{amsbsy}\usepackage{epsfig}
\usepackage{subcaption}
\floatstyle{ruled}
\newfloat{algorithm}{tbp}{loa}
\providecommand{\algorithmname}{Algorithm}
\floatname{algorithm}{\protect\algorithmname}
\newcommand{\beq}{\begin{equation}}
\newcommand{\eeq}{\end{equation}}

\def\bu{\mbox{\boldmath $u$}}

\def\by{\mbox{\boldmath $y$}}

\def\bs{\mbox{\boldmath $s$}}

\def\by{\mbox{\boldmath $y$}}

\def\bY{\mbox{\boldmath $Y$}}

\def\bmu{\mbox{\boldmath $\mu$}}

\def\bP{\mbox{\boldmath $P$}}

\def\mSigma{\mbox{$\mathbf{\Sigma}$}}
\def\my{\mbox{$\mathbf{y}$}}
\def\me{\mbox{$\mathbf{e}$}}

\def\mx{\mbox{$\mathbf{x}$}}
\def\mb{\mbox{$\mathbf{b}$}}
\def\mA{\mbox{$\mathbf{A}$}}
\def\mB{\mbox{$\mathbf{B}$}}

\def\mC{\mbox{$\mathbf{C}$}}
\def\mLambda{\mbox{$\mathbf{\Lambda}$}}
\def\mlambda{\mbox{$\mathbf{\lambda}$}}
\def\mPhi{\mbox{$\mathbf{\Phi}$}}
\def\mX{\mbox{$\mathbf{X}$}}
\def\mD{\mbox{$\mathbf{D}$}}
\def\mE{\mbox{$\mathbf{E}$}}

\def\mF{\mbox{$\mathbf{F}$}}
\def\mH{\mbox{$\mathbf{H}$}}
\def\mI{\mbox{$\mathbf{I}$}}
\def\mL{\mbox{$\mathbf{L}$}}
\def\mY{\mbox{$\mathbf{Y}$}}
\def\mP{\mbox{$\mathbf{P}$}}
\def\mQ{\mbox{$\mathbf{Q}$}}
\def\mZ{\mbox{$\mathbf{Z}$}}

\def\mS{\mbox{$\mathbf{S}$}}

\def\mU{\mbox{$\mathbf{U}$}}
\def\mV{\mbox{$\mathbf{V}$}}
\def\mz{\mbox{$\mathbf{z}$}}
\def\mv{\mbox{$\mathbf{v}$}}
\def\mW{\mbox{$\mathbf{W}$}}
\def\mM{\mbox{$\mathbf{M}$}}

\newcommand{\ds}{\displaystyle}

\floatstyle{ruled}
\newfloat{algorithm}{tbp}{loa}
\providecommand{\algorithmname}{Algorithm}
\floatname{algorithm}{\protect\algorithmname}

\newtheorem{proposition}{Proposition}

\newenvironment{proof}[1][Proof]{\noindent \textbf{#1.} }{\qedsymbol}
\newcommand{\qedsymbol}{\hspace{\fill}\rule{1.5ex}{1.5ex}}

\begin{document}
\title{ Graph topology inference based on \\ sparsifying transform learning}
\author{Stefania Sardellitti,~\IEEEmembership{Member,~IEEE}, Sergio Barbarossa,~\IEEEmembership{Fellow,~IEEE}, and Paolo Di Lorenzo,~\IEEEmembership{Member,~IEEE} \\
\thanks{The authors are with the Department of Information
Engineering, Electronics, and Telecommunications,
Sapienza University of Rome, Via Eudossiana 18, 00184,
Rome, Italy. E-mails: \{stefania.sardellitti; sergio.barbarossa; paolo.dilorenzo\}@uniroma1.it.
This work has been
supported by
H2020 EUJ Project 5G-MiEdge, Nr. 723171.
This paper was presented in part at the  IEEE Global Conference on Signal and Information Processing, GlobalSIP,  USA, Dec. 7-9, 2016.
}}
\vspace{-0.4cm}

\maketitle
\begin{abstract}
Graph-based representations play a key role in machine learning. The fundamental step in these representations is the association of a graph structure to a dataset. In this paper, we propose a method that aims at finding a block sparse representation of the graph signal leading to a modular graph whose Laplacian matrix admits the found dictionary as its eigenvectors. The role of sparsity here is to induce a band-limited representation or, equivalently, a modular structure of the graph.
The proposed strategy is composed of two optimization steps: i) learning an orthonormal sparsifying transform from the data; ii) recovering the Laplacian, and then topology, from the transform. The  first step is achieved through an iterative algorithm whose alternating intermediate solutions are expressed in closed form. The second step recovers the Laplacian matrix from the sparsifying transform through a convex optimization method. 
Numerical results corroborate the effectiveness of the proposed methods over both synthetic data and real brain data, used for inferring the brain functionality network through experiments conducted over patients affected by epilepsy. \end{abstract}
\section{Introduction}
\label{sec:intro}
In many machine learning applications, from brain functional analysis to gene regulatory networks or sensor networks, associating a graph-based representation to a dataset is a fundamental step to extract relevant information from the data, like detecting clusters or measuring node centrality. The graph underlying a given data-set can be real, as in the case of a sensor network, for example, or it can be just an abstract descriptor of pairwise similarities.
There is  a large amount of work aimed at learning the network topology from a set of observations \cite{Kolaczyk}. Typically, the graph topology reflects correlations among signals defined over its vertices. However, looking only at correlations may fail to capture the causality relations existing among the data. Alternative approaches built on partial correlation \cite{Kolaczyk}, or Gaussian graphical models \cite{Tibshirani}, \cite{Tenenbaum} have been deeply investigated. The analysis of signal defined over graphs, or Graph Signal Processing (GSP)  \cite{Shuman2013, Moura2014, Pesenson2008, Zhu2012}, has provided a further strong impulse to the discovery of new techniques for inferring the graph topology from a set of observations.
Some GSP-based approaches make assumptions about the graph by enforcing properties such as sparsity and/or smoothness of the signals
\cite{Kalofolias}, \cite{Hero}, \cite{Frossard_jour}.  Specifically, Kalofolias in \cite{Kalofolias}
 proposed  a Laplacian learning framework,  for smooth graph signals, by reformulating the problem as a weight $\ell_1$-norm minimization   with a log barrier penalty term on the node degrees.
 Recently, the problem of learning a sparse, unweighted, graph Laplacian from smooth graph signals has been modeled by Chepuri {\it et al.}   as a rank ordering problem, where the  Laplacian is expressed in terms of a sparse edge selection vector by assuming a priori knowledge of the graph sparsity level \cite{Hero}.
An alternative approach consists in defining joint properties between signals and underlying graph such that the signal representation is consistent with given
statistical priors on latent variables \cite{Frossard_jour}, \cite{Egilmez}. Under the assumption that the observed field is a Gaussian Markov Random Field (GMRF), whose precision matrix is the graph Laplacian, Dong \emph{et al.}  \cite{Frossard_jour} proposed a graph learning method that favors signal representations that are smooth and consistent with the underlying statistical prior model.
Recently, in \cite{Egilmez} Egilmez {\it et al.} proposed a general framework where graph learning is formulated as the estimation of different types of graph Laplacian matrices by adopting as minimization metric  the log-determinant Bregman divergence with an additional, sparsity promoting, $\ell_1$-regularization term.  In \cite{Egilmez} it is shown that such a problem corresponds to a maximum a posteriori parameter estimation of GMRFs whose precision matrices are graph Laplacians.
In case of directed graphs, Mei {\it et al.}  proposed algorithms for estimating the directed, weighted adjacency matrices describing the dependencies among time series generated by systems of interacting agents \cite{Mei2015}.

There are also some recent works that focus on learning the graph topology from  signals that diffuse over a graph \cite{Segarrajournal2016}, \cite{Pasdeloup},
\cite{Thanou}. Segarra \emph{et al.} developed efficient sparse graph recovery algorithms for identifying  the  graph shift operator (the adjacency matrix or the normalized graph Laplacian) given only the eigenvectors  of the shift operator  \cite{Segarrajournal2016}.
These bases are estimated from the eigenvectors of the sample covariance matrix of stationary graph signals resulting from diffusion dynamics on the graphs.
Then the optimization strategy minimizes the shift-operator $\ell_1$-norm by imposing complete or partial knowledge of the eigenvectors of the empirical covariance matrix. In \cite{Pasdeloup}, Pasdeloup {\it et al.} formulated the graph-inference problem in a way similar  to  \cite{Segarrajournal2016}, albeit with a different matrix selection strategy,  which aims at  recovering   a  matrix modeling a diffusion process without imposing the structure of a valid Laplacian to this matrix.
Recently, Thanou {\it et al.} proposed a graph  learning strategy based on the assumption that the graph signals are generated from heat diffusion processes starting from a few nodes of the graphs, by enforcing sparsity of the graph signals in a dictionary that is a concatenation of graph diffusion kernels at different scales \cite{Thanou}.
The problem of jointly estimating the dictionary matrix and the graph topology from data has been recently tackled by Yankelevsky {\it et al.}  by taking into account the underlying structure in both the signal and the manifold domains \cite{Elad_Yank}.
Specifically, two smoothness regularization terms are introduced to describe, respectively, the relations between the rows and the columns of the data matrix.
This leads to a dual graph regularized formulation aimed at finding, jointly, the dictionary and the sparse data matrix.
The   graph Laplacian is then estimated by imposing  smoothness of any signal represented over the estimated dictionary.
Given the observed graph signal, the estimation of the Laplacian eigenvectors is  conceptually equivalent to finding a  sparsifying transform \cite{Bresler13}.
A dictionary learning method for sparse coding incorporating a graph Laplacian regularization term, which enables the exploitation of data-similarity, was developed in \cite{BMVC2015_44}, while a dictionary learning method reinforcing group-graph structures was proposed in \cite{AAAI159476}.

One of the key properties of the spectral representation of signals defined over a graph is that the representation depends on the graph topology. The Graph Fourier Transform (GFT), for example, is the projection of a signal onto the space spanned by the eigenvectors of the graph Laplacian matrix. The main idea underlying our work is to associate a graph topology to the data in order to make the observed signals band-limited over the inferred graph or, equivalently, to make their spectral representation {\it sparse}. Finding this sparse representation of a signal on a graph is instrumental to finding graph topologies that are {\it modular}, i.e. graphs characterized by densely interconnected subsets (clusters) weakly interconnected with other  clusters. Modularity is in fact a structure that conveys important information. This is useful in all those applications, like clustering or classification, whose goal is to associate a label to all vertices belonging to the same class or cluster. If we look at the labels as a signal defined over the vertices of a graph, we have a  signal  constant within each cluster and varying arbitrarily across different clusters. This is a typical band-limited signal over a graph.  The proposed approach is composed of the following two main steps:
i) learn the GFT basis and the sparse signal representation  {\it jointly} from the observed dataset;
ii) infer the graph \textit{weighted} Laplacian, and then the graph topology, from the estimated Laplacian eigenvectors.
A nice feature of the proposed approach is that the first step is solved through an alternating optimization algorithm whose single steps are solved in closed form. Furthermore, the second step is cast as a convex problem, so that it can be efficiently solved.
With respect to previous works, we make no assumptions about the stationarity  of the observed signal and we do not assume any diffusion process taking place on the graph.

Interestingly, we find theoretical conditions relating the sparsity of the graph that can be inferred and the bandwidth (sparsity level) of the signal.

To assess the performance of the proposed approaches, we consider both synthetic and real data. In the first case, we know the ground-truth graph and then we use it as a benchmark. In such a case, we also compare our methods with similar approaches, under different data models (e.g., band-limitedness, multivariate Gaussian random variables described by an inverse Laplacian covariance matrix, etc.). Then, we applied our methods to infer the brain functionality graph from electrocorticography (ECoG) seizure data collected in an epilepsy study \cite{Kolaczyk}, \cite{Kramer}. In such a case, the brain network is not directly observable and must be inferred by measurable ECoG signals. This represents an important application for the recently developed, graph-based learning methods \cite{Giannakis},\cite{Ribeiro}.  By exploiting such methods, the  hope is to gain more insights about the unknown mechanisms underlying some
neurological diseases such as, for example, epileptic seizures.

The rest of the paper is organized as follows.
Section II introduces the graph signal models and the topology learning strategy. Section III describes the first step of our method designed to
learn jointly, a subset of the GFT basis and the sparse signal.  Two alternative algorithms to learn the Laplacian matrix from the estimated (incomplete) Fourier basis are then
presented in Section IV. Theoretical conditions on graph identifiability are discussed in Section V.  Numerical tests based on both synthetic and real data from an epilepsy study, are then presented in Section VI.


\section{Learning  topology from graph signals}
\subsection{Graph signals model}
We consider an undirected graph $\mathcal{G}=\{\mathcal{V},\mathcal{E}\}$
consisting of  a set of $N$ vertices (or nodes) $\mathcal{V}=\{1,\ldots, N\}$ along with a set of edges
$\mathcal{E}=\{a_{ij}\}_{i,j \in \mathcal{V}}$, such that  $a_{ij}>0$, if there is a  link between node $j$
and node $i$, or $a_{ij}=0$, otherwise. We denote with $|\mathcal{V}|$
the cardinality of $\mathcal{V}$, i.e. the number of elements of $\mathcal{V}$.
Let $\mA$ denote the $N\times N$ adjacency, symmetric matrix with entries the edge weights $a_{ij}$ for $i,j=1,\ldots,N$.
The (unnormalized) graph Laplacian is defined as $\mL:= \mD-\mA$, where the  degree matrix $\mD$ is a diagonal matrix
whose $i$th diagonal entry is $d_i= \sum_j a_{ij}$. Since $\mathcal{G}$ is an undirected graph,
the associated Laplacian matrix is symmetric and positive semidefinite and admits the eigendecomposition
$\mL= \mU \mLambda \mU^T$, where $\mU$ collects all the eigenvectors $\{\bu_i\}_{i=1}^{N}$ of $\mL$, whereas $\mLambda$ is a diagonal matrix
containing the eigenvalues of $\mL$.
A signal $\by$ on a graph $\mathcal{G}$ is defined as a mapping from the vertex set
to a complex vector of size $N=|\mathcal{V}|$, i.e. $\by: \mathcal{V}\rightarrow \mathbb{C}$.
For undirected graphs, the GFT $\hat{\by}$ of a graph signal $\by$ is
 defined as the projection of $\by$ onto the subspace spanned by the eigenvectors of the Laplacian matrix, see e.g. \cite{Shuman2013}, \cite{Pesenson2008}, \cite{Zhu2012}:
 \beq
 \hat{\by}=\mU^T \by.
  \eeq
A band-limited graph signal is a signal whose GFT is sparse. Let us denote these signals as
  \begin{equation}
\label{y=Vs}
\by=\mU \bs
\end{equation}
where $\bs$ is a sparse vector.
One of the main motivations for projecting the
signal $\by$ onto the subspace spanned by the eigenvectors of
$\mL$, is that these eigenvectors encode some of the
graph topological features \cite{Luxburg}, \cite{Chung1997}. In particular, the eigenvectors associated to the smallest eigenvalues of $\mL$
identify graph clusters  \cite{Luxburg}. Hence, a band-limited signal with support on the lowest indices is a signal that is smooth within clusters, while it can vary arbitrarily across clusters.
More formally, given a subset of indices $\mathcal{K}\subseteq \mathcal{V}$,
 we introduce the {\it band-limiting} operator $\mB_{\mathcal{K}}=\mU \mSigma_{\mathcal{K}} \mU^T$,
 where $\mSigma_{\mathcal{K}}$ is a diagonal matrix defined as
 $\mSigma_{\mathcal{K}}=\mbox{Diag}\{\mathbf{1}_{\mathcal{K}}\}$ and $\mathbf{1}_{\mathcal{K}}$ is the set indicator vector, whose $i$-th entry is equal to one if $i \in \mathcal{K}$, and zero otherwise.
 \\
 Let us assume that $\bs \in \mathbb{R}^N$ is a  $|\mathcal{K}|$-sparse vector, i.e.  $\parallel\bs\parallel_0 \leq K $, where the $l_0$-norm counts the number
of non-zeros in $\bs$ and $K=|\mathcal{K}|$.
 We say that  $\by$ is  a  $\mathcal{K}$-band-limited graph signal if
$\mB_{\mathcal{K}}\by=\by$, or,  equivalently, if $\by=\mU \bs$, where $\bs$ is a
$|\mathcal{K}|$-sparse vector \cite{Tsit_Barb_PDL}.
The observed signal $\by \in \mathbb{R}^N$  can be seen as  a linear combination of columns from a dictionary $\mU \in
\mathbb{R}^{N\times N}$.
Collecting $M$ vectors $\bs_i \in \mathbb{R}^N$
 in the matrix $\mS \in \mathbb{R}^{N\times M}$, we assume that the signals $\bs_i$ share a common support. This implies that $\mS$ is a block-sparse matrix
 having entire rows as zeros or nonzeros.  We define the set of $K$-block sparse signals \cite{Baraniuk}, \cite{Eldar} as
 \beq
 \begin{split}
 \mathcal{B}_K =\{ \mS=[\bs_1,\ldots, \bs_M] \in \mathbb{R}^{N\times M} \, | \,   \mS(i,:)=\mathbf{0},\; & \forall\, i \not\in  \mathcal{K}\subseteq \mathcal{V},\\ &  \, K=|\mathcal{K}| \}
 \end{split}
 \eeq
 where $\mS(i,:)$ denotes the $i$th row of $\mS$.
Hence, by introducing the  matrix $\mY \in \mathbb{R}^{N\times M}$    with columns the $M$ observed vectors $\by_i$,
  we get the compact form $\mY=\mU \mS$
  where    $\mS \in \mathcal{B}_K$.\\
  \subsection{Problem formulation}
  In this paper we propose  a  novel strategy to learn,  jointly, the orthonormal transform matrix $\mU$, the sparse graph signal $\mS$ and
  the underlying graph topology $\mL$. Although  the matrices $\mU$,
$\mS$ and $\mL$ are not uniquely identified, we will devise an efficient algorithm to  find optimal solutions that  minimize the estimation error.
Before formulating our problem, we  introduce the space of valid combinatorial Laplacians, i.e.
\begin{equation}
\label{L}
\mathcal{L}=\{  \mL \in \mathcal{S}_{+}^N \, |   \; \mL \mathbf{1}=\mathbf{0},\; L_{ij}=L_{ji} \leq 0, \forall i\neq j\}
\end{equation}
 where $\mathcal{S}_{+}^N$ is the set of  real, symmetric and positive semidefinite matrices.
In principle, our goal would be to find the solution of  the following optimization problem
\beq \nonumber \label{eq:P_tot}
\begin{split}
 & \underset{\begin{split} & \mL, \mU,\mLambda \in \mathbb{R}^{N \times N}\\ & \mS \in \mathbb{R}^{N \times M} \end{split}}{\hspace{-1cm}\min}  \; \parallel  \mY- \mU \mS\parallel^2_{F} +f(\mL,\mY,\mS)\hspace{1.5cm} (\mathcal{P})\medskip\\
 & \hspace{0.4cm}  \begin{array}{cll}    \mbox{s.t.}
   & \hspace{1.2cm} \text{a}) \; \mU^T \mU=\mI, \quad \bu_1=b \mathbf{1} \medskip \\
 &  \hspace{1.2cm} \text{b}) \;\mL\mU=\mU \mLambda, \; \, \mL \in \mathcal{L}, \; \text{tr}(\mL)=p \medskip\\
&  \hspace{1.2cm} \text{c}) \; \mLambda\succeq \mathbf{0}, \Lambda_{ij}=\Lambda_{ji}=0, \forall i \neq j \medskip\\
 & \hspace{1.2cm} \text{d}) \; \mS \in \mathcal{B}_K, \medskip\\
 \end{array}
   \end{split}
\eeq
where: the objective function is the sum of the data fitting  error plus the function $f(\mL,\mY,\mS)$ enabling, as we will see in Section IV,  the recovering of a graph topology which reflects the smoothness of the observed graph signals; the constraint a) forces $\mU$ to be unitary and includes a vector proportional to the  vector of all ones; b) constrains $\mL$ to possess the desired structure and imposes that $\mU$ and $\mLambda$ are its eigen-decomposition; the trace constraint, with $p>0$, is used to avoid the trivial solution; c)  forces $\mL$ to be positive semidefinite  (and $\mLambda$ to be diagonal); finally,
d) enforces the $K$-block sparsity of the signals. However, problem $\mathcal{P}$ is  non-convex in both the objective function and the constraint set.
To make the problem tractable,  we devise a strategy that decouples   $\mathcal{P}$ in two simpler optimization problems.
The proposed strategy is rooted on two basic steps: i) first, given the observation matrix $\mY$, learn  the orthonormal  transform matrix $\mU$ and the sparse signal matrix $\mS$ {\it jointly}, in order to minimize the fitting error $\parallel\mY-\mU\mS\parallel_{\text{F}}^{2}$, subject to a block-sparsity constraint; ii) second, given $\mU$,  infer the graph Laplacian matrix $\mL$ that admits the column of $\mU$ as its eigenvectors.
While  the theoretical proof of convergence of this double-step approach is still an open question, the effectiveness of the proposed strategy has been corroborated by extensive numerical results over both simulated and real data.
In the next sections, we will develop alternative, suboptimal algorithms to solve  the proposed graph inference problem and we will present the numerical results.

 \section{LEARNING FOURIER BASIS AND GRAPH SIGNALS JOINTLY}
In this section we describe   the first step of the proposed graph recovering strategy, aimed at learning the pair of matrices $(\mU,\mS)$  jointly, up to a multiplicative rotation matrix. Observe that, since $\mU \in \mathbb{R}^{N\times N}$ is full rank, if $\mU$ is known, recovering $\mS$ from $\mY$ is easy.
However, in our case, $\mU$ is unknown. Hence, recovering {\it both} $\mS$ and $\mU$ is  a challenging
task, which is prone to an identifiability problem.
To solve this task,  we use a method conceptually similar to sparsifying transform learning, assuming a unitary transform \cite{Bresler15}, with the only difference that we enforce {\it block}-sparsity and we impose, a priori, that one of the eigenvectors be a constant vector. More specifically, we set $\bu_1=b \mathbf{1}$, with $b=1/\sqrt{N}$.

Because of the unitary property of $\mU$, we can write $\parallel \mY-\mU \mS\parallel^2_{F}=\parallel \mU^T \mY- \mS\parallel^2_{F}$. Our goal is to find the block-sparse columns $\{\bs_i\}_{i=1}^M$, and the Fourier basis  orthonormal vectors $\{\bu_i\}_{i=1}^N$,  minimizing the sparsification error.
 Therefore, given the training matrix $\mY \in \mathbb{R}^{N \times M}$, motivated by the original problem $\mathcal{P}$, we formulate the following optimization problem:
 \beq \label{eq:P}
\begin{split}
 & \underset{\mathbf{U} \in \mathbb{R}^{N \times N}, \mathbf{S} \in \mathbb{R}^{N \times M}}{\hspace{-1cm}\min}  \quad \parallel \mU^T \mY- \mS\parallel^2_{F} \hspace{1.5cm} (\mathcal{P}_{U,S})\medskip\\
 & \hspace{0.4cm}  \begin{array}{cll}    \mbox{s.t.}   & \hspace{1.8cm}  \mU^T \mU=\mI, \quad \bu_1=b \mathbf{1}  \medskip \\  & \hspace{1.8cm}  \mS \in \mathcal{B}_K.
 \end{array}
   \end{split}
\eeq
 Although the objective function is convex, problem $\mathcal{P}_{U,S}$ is non-convex due to the nonconvexity of both the orthonormality and sparsity
constraints.
In the following, inspired by  \cite{Bresler15}, we propose an  algorithm to solve $\mathcal{P}_{U,S}$  that alternates between solving for the block-sparse signal matrix $\mS$ and for the orthonormal  transform  $\mU$.
\subsection{Alternating optimization algorithm}
The proposed approach to solve the above non-convex problem $\mathcal{P}_{U,S}$ alternates between the minimization of the objective function  with respect to $\mS$ and $\mU$ at each step $k$,  iteratively,
as follows:
\beq \nonumber
\begin{array}{lll}
    \begin{array}{ll}  \text{1.} \quad \hat{\mS}^k \triangleq & \underset{\mathbf{S} \in \mathbb{R}^{N \times M}}{\arg \min}  \quad \| (\hat{\mU}^{k-1})^{T} \mY-\mS\|^{2}_{F}
    \hspace{0.9cm} (\mathcal{S}_k) \\
     & \hspace{0.4cm} \mbox{s.t.}  \hspace{0.8cm} \mS \in \mathcal{B}_K
 \end{array} \medskip \vspace{0.3cm}\\
    \begin{array}{ll}  \text{2.} \quad  \hat{\mU}^k \triangleq & \underset{\mathbf{U} \in \mathbb{R}^{N \times N}}{\arg \min}  \quad \| \mU^{T} \mY-\hat{\mS}^{k}\|^{2}_{F} \hspace{1.4cm} (\mathcal{U}_k)\medskip\\
    & \hspace{0.4cm} \mbox{s.t.} \hspace{0.8cm} \mU^T\mU=\mI,  \bu_1=b \mathbf{1}. \end{array}\medskip
 \end{array} \label{iter_alg}
\eeq
The algorithm iterates until a termination criterion is met, within a prescribed accuracy, as summarized in Algorithm \ref{algorithm:Alg_1}.

\begin{algorithm}[t]
\small

    \quad  {Set}   $\mU^0 \in \mathbb{R}^{N \times N}$, $\mU^{0 \, T} \mU^0=\mI$, $\bu_1=b \mathbf{1}$, $k=1$ 

    \quad  {\textbf{Repeat}}

 \quad \quad           {Find} $\hat{\mS}^k$ {as} {solution} {of} $\mathcal{S}_k$ in \,(\ref{eq:S_iter});

   \quad \quad          {Compute} {the} {optimal point} $\hat{\mU}^k$ as in  (\ref{U_op});

  \quad \quad         $k=k+1$;

  \quad {\textbf{until}}   {\textbf{convergence}
  }

       \caption{: Learning  Fourier basis and graph signals}
 \label{algorithm:Alg_1}
\end{algorithm}

The interesting feature is that  the two non-convex problems $\mathcal{S}_k$ and $\mathcal{U}_k$
admit closed-form solutions.  Therefore,  both steps of the above algorithm can be performed exactly, with a
low computational cost, as follows.\\
\textit{1) Computation of the critical point $\hat{\mS}^k$:}
To solve problem $\mathcal{S}_k$ we need to define the $(p,q)$-mixed norm of the matrix $\mS$ as
\beq
\parallel \mS \parallel_{(p,q)}={\left(\ds \sum_{i=1}^{N} \parallel \mS(i,:)\parallel_{p}^q \right)}^{1/q}.
\eeq
When $q=0$, $\parallel \mS \parallel_{(p,0)}$ simply counts the number of nonzero rows in $\mS$, whereas for $p=q=2$ we get
 $\parallel \mS \parallel_{(2,2)}^2=\parallel \mS \parallel_{F}^2$.
  Hence, we reformulate problem $\mathcal{S}_k$ as in \cite{Baraniuk}, to obtain the best block-sparse approximation of the matrix $\mS$ as follows
\beq  \label{eq:S_iter}
\begin{array}{lll}
    \begin{array}{ll} & \underset{\mathbf{S} \in \mathbb{R}^{N \times M}}{\arg \min}  \quad \| (\hat{\mU}^{k-1})^{T} \mY-\mS\|^{2}_{F}
    \hspace{0.9cm} (\mathcal{S}_k) \\
     & \hspace{0.4cm} \mbox{s.t.}  \hspace{0.7cm} \parallel \mS \parallel_{(2,0)}\leq K
\end{array}\end{array}
 \eeq
 where   we force the rows of $\mS$   to be $K$-block sparse.
The solution of this problem is simply obtained by sorting the rows of $(\hat{\mU}^{k-1})^{T} \mY$ by their $l_2$-norm and then selecting the rows with largest norms \cite{Baraniuk}.

 \textit{2) Computation of the critical point $\hat{\mU}^k$:} The optimal transform update  for the transform matrix $\mathbf{U}^k$ is provided in the following proposition.
\begin{proposition}
Define $\bar{\mY}^{k}= \mY (\mS^{k})^T$, ${\mZ}^{k}=[\bar{\my}_2^{k},\ldots,\bar{\my}_N^{k}]$,  where $\bar{\my}_i^{k}$ represents the $i$th column of  $\bar{\mY}^{k}$, and
  $\mZ_{\perp}^{k}=\mP \mZ^{k}$
with $\bP=\mI-\mathbf{1} \mathbf{1}^T/N$.
If $r=\text{rank}({\mZ}_{\perp}^{k})$, we have
\beq \label{Zp}
{\mZ}_{\perp}^{k}=\mX  \mSigma \mV^T=[\mX_r \, \mX_s]  \mSigma \mV^T
\eeq
where the $N\times r$ matrix $\mX_r$ contains as columns  the  $r$
left-eigenvectors associated to the non-zero singular values of $\mZ_{\perp}^{k}$, $\mX_s$ is an orthonormal  $N\times N-r$ matrix selected  as  belonging
to the nullspace of the matrix $\mB=[ \mathbf{1}^T; \mX_r^T]$, i.e. $\mB \mX_s=\mathbf{0}$.
The optimal solution of the non-convex problem $\mathcal{U}_k$ is then
\beq \label{U_op}
\hat{\mU}^{k}=[ \mathbf{1}b \quad \bar{\mU}^{\star}]
\eeq
where $\bar{\mU}^{\star}=\mX^{-} \mV^T$ and $\mX^{-}$ is obtained by removing from $\mX$ its last column.
\end{proposition}
\begin{proof}
 See Appendix A.
 \end{proof}

Of course, the pair of matrices $(\mU, \mS)$ solving (\ref{eq:P}) can only be found up to the multiplication by a unitary matrix $\mW$ that preserves the sparsity of $\mS$. In fact,
$\mU\mS=\mU \mW \mW^T\mS:= \hat{\mU}\hat{\mS}$, where $\mW$ is a unitary matrix, i.e., $\mW\mW^T=\mI$. In the next section, we will take into account the presence of this unknown rotation matrix in the recovery of the graph topology. Nevertheless,  if the graph signals assume values in a discrete alphabet, it is possible to de-rotate the estimated transform $\hat{\mU}$ by using the adaptive, iterative algorithm developed in \cite{Macchi}, based on the stochastic gradient approximation, as shown in Section VI.

\section{LEARNING THE GRAPH TOPOLOGY}
In this section we propose a strategy aimed at identifying the topology of the graph  by learning the Laplacian matrix $\mL$
from the estimated transform matrix $\hat{\mU}$.
Of course, if we would know the eigenvalue and eigenvector matrices of $\mL$, we could find $\mL$ by solving
\beq \mL \mU= \mU \mLambda.
\eeq
However, from the algorithm solving problem (\ref{eq:P}) we can only expect to find two matrices $\hat{\mU}=\mU \mW$ and $\hat{\mS}=\mW^T \mS$, where $\mW$ is a unitary matrix that preserves the block-sparsity of $\mS$. This means that, if we take the product $\mL \hat{\mU}$, we get
 \beq
\mL \hat{\mU}= \mL \mU \mW= \mU \mW \mW^{T}\mLambda \mW.
 \eeq
Introducing the matrix $\mC=\mW^{T}\mLambda \mW\succeq \mathbf{0}$, we can write
 \beq \label{eq:C_const}
 \mL \hat{\mU}= \hat{\mU} \mC.
  \eeq
Furthermore, having supposed that the observed graph signals are $\mathcal{K}$-bandlimited,
with $|\mathcal{K}|=K$, we can only assume, from the first step of the algorithm, knowledge of the $K$ columns of $\hat{\mU}$ associated to the observed signal subspace.
Under this setting, the constraint in  (\ref{eq:C_const}) reduces to
 \beq \label{eq:C_const2}
 \mL \hat{\mU}_{\mathcal{K}}= \hat{\mU}_{\mathcal{K}} \mC_{\mathcal{K}}
  \eeq
 with $\mC_{\mathcal{K}}\in \mathbb{R}^{K\times K}$, $\mC_{\mathcal{K}}=\mPhi^{-1}\mLambda_{\mathcal{K}} \mPhi\succeq \mathbf{0}$, $\mPhi$ a unitary matrix and $\mLambda_{\mathcal{K}}$ the diagonal matrix obtained by selecting the set $\mathcal{K}$ of diagonal entries of $\mLambda$.

For the sake of simplifying the optimization problem $\mathcal{P}$, the graph learning problem can be formulated as follows
\beq \nonumber
\begin{array}{lll}
\underset{\mathbf{L} \in \mathbb{R}^{N \times N}, \mathbf{C}_{\mathcal{K}} \in \mathbb{R}^{K \times K}}{\min} \quad \!  f(\mL, \mY,\hat{\mS})\qquad \qquad \qquad(\mathcal{P}_{f})\\
\hspace{1cm} \left. \begin{array}{lll} \mbox{s.t.}
      &\hspace{0.8cm}\mL \in \mathcal{L}, \text{tr}(\mL)=p \medskip\\
       &\hspace{0.8cm}\mL \hat{\mU}_{\mathcal{K}}= \hat{\mU}_{\mathcal{K}} \mC_{\mathcal{K}}, \; \mC_{\mathcal{K}}\succeq \mathbf{0}
      \end{array} \right \} \triangleq \mathcal{X}(\hat{\mU}_{\mathcal{K}})
\end{array}
\eeq
with $\mathcal{L}$ defined as in (\ref{L}). Different choices for
the objective function $f(\mL,\mY, \hat{\mS})$ have been proposed in the literature, such as, for example: i) $f(\mL)=||\mL||_0$, as  in  \cite{Segarrajournal2016}, or
ii) $f(\mL,\mY)=\text{tr}(\mY \mL \mY^T)+\mu ||\mL||_F$, as in \cite{Frossard_jour}. In the first case, the goal was to recover the sparsest graph, given the  eigenvectors of the shift matrix, estimated via principal component analysis (PCA) of  the sample covariance of graph signals diffused on the  network. In the second case, the goal was to minimize a linear combination of the $l_2$-norm total variation of the observed signal $\mY$, measuring the signal smoothness, plus a Frobenius norm term, specifically added to control (through the coefficient $\mu\geq 0$) the distribution of the off-diagonal
entries of $\mL$. Note that, because of the constraint on the Laplacian trace, setting $\mu$ to large values penalizes large degrees and
leads, in the limit, to dense graphs with constant degrees across the nodes.\\
In this paper, our goal is to infer a graph topology that gives rise to a band-limited signal, consistent with the observed signal. This inference is based on  the joint estimation of the $|{\cal K}|$-sparse signal vector $\bs_i$ and the Laplacian eigenvectors \textit{jointly}. To this end, we propose two alternative  choices for the objective function, leading to the following two algorithms.
\subsection{Total variation based graph learning}
The  Total Variation Graph Learning (TV-GL)  algorithm, is formulated  as follows
\beq \nonumber
\begin{array}{ll}
\underset{\substack{\mathbf{L} \in \mathbb{R}^{N \times N} \medskip \\ \; \; \mathbf{C}_{\mathcal{K}} \in \mathbb{R}^{K \times K}}}{\min} \quad f_1(\mL,\mY,\mu)\triangleq   \text{TV}(\mL,\mY) + \mu \parallel \mL \parallel_{F}^{2} \hspace{0.55cm} (\mathcal{P}_{f_1})\\
\hspace{0.4cm}  \begin{array}{lll} \mbox{s.t.}
      \hspace{1cm}(\mL, \mC_{\mathcal{K}})  \in \mathcal{X}(\hat{\mU}_{\mathcal{K}})     \end{array}
\end{array}
\eeq
where $\text{TV}(\mL,\mY)=\sum_{m=1}^{M}\sum_{i,j=1}^{N} L_{i j} \mid Y_{i m}- Y_{j m}\mid$ represents the $\ell_1$-norm total variation of the observed signal $\mY$ on the graph. Minimizing this $\ell_1$-norm tends to disconnect nodes having distinct signal values and to connect nodes have similar values, whereas the Frobenius norm controls the sparsity of the off-diagonal entries of $\mL$; the constraint on the Laplacian trace in $\mathcal{X}(\hat{\mU}_{\mathcal{K}})$ is instrumental to avoid the trivial solution by fixing the $l_1$-norm
 of $\mL$. Note that this constraint can be read as a sparsity constraint on $\mL$, since  $\text{tr}(\mL)=\parallel \mL\parallel_1$, and  the $l_1$ norm represents the tightest relaxation of the $l_0$ norm.
The coefficient $\mu$ is used to control the graph sparsity: if $\mu$ increases the Frobenius norm of $\mL$ tends to decrease and, because of the trace constraint, this leads to a more uniform distribution of the off-diagonal entries so that the number of edges increases; on the contrary, as $\mu$ goes to $0$,
 the graph tends  to be more and more sparse.

Problem $\mathcal{P}_{f_1}$ is convex since both feasible set  and objective function  are convex.

 \subsection{Estimated-signal-aided graph learning}
 We propose now an alternative method, which we name Estimated-Signal-Aided Graph Learning (ESA-GL) algorithm. In this case the  signals $\hat{\mS}$ estimated in Algorithm $1$ are used in the graph recovering method. To motivate this method,  we  begin observing that
  \beq \label{eq:l2_norm_var}
 \text{tr}(\mY^T \mL \mY)=\text{tr}({\mS}^T_{\mathcal{K}} \mLambda_{\mathcal{K}}  {\mS}_{\mathcal{K}})=\text{tr}(\hat{\mS}^T_{\mathcal{K}} \mC_{\mathcal{K}} \hat{\mS}_{\mathcal{K}})
 \eeq
  where $\mS_{\mathcal{K}}$  contains the rows with index in the support set of the  graph signals, whereas
  $\hat{\mS}_{\mathcal{K}}=\mPhi^T \mS_{\mathcal{K}}$ are  the graph signals  estimated by Algorithm \ref{algorithm:Alg_1}.
  Therefore, an alternative formulation of problem $\mathcal{P}_f$ aiming at minimizing the
 smoothing term  in
(\ref{eq:l2_norm_var}),
   is
    \beq \nonumber
\begin{array}{ll}
\underset{\substack{\mathbf{L} \in \mathbb{R}^{N \times N} \medskip \\ \; \; \mathbf{C}_{\mathcal{K}} \in \mathbb{R}^{K \times K}}}{\min} \quad f_2(\mL,\hat{\mS}_{\mathcal{K}},\mu)\triangleq \text{tr}(\hat{\mS}^T_{\mathcal{K}} \mC_{\mathcal{K}} \hat{\mS}_{\mathcal{K}})+\mu ||\mL||_F \hspace{0.4cm} (\mathcal{P}_{f_2})\\
\hspace{0.5cm}  \begin{array}{lll} \mbox{s.t.}
      \hspace{0.8cm}(\mL, \mC_{\mathcal{K}})  \in \mathcal{X}(\hat{\mU}_{\mathcal{K}}).     \end{array}
\end{array}
\eeq
Also in this case the formulated problem is convex, so that it can be efficiently solved.

  \section{Laplacian identifiability conditions}\label{sec:identif}
Before assessing the goodness of the proposed inference algorithms, in this section we investigate the conditions under which the Laplacian matrix can be recovered uniquely. Interestingly, the derivations lead to an interesting relation between the bandwidth of the signal and the sparsity of the graph. We start looking for the conditions under which the feasible set $\mathcal{X}(\hat{\mU}_{\mathcal{K}})$ of problems $\mathcal{P}_{f_1}$ and  $\mathcal{P}_{f_2}$ reduces to a singleton.
To  derive closed form bounds for the identifiability conditions of the Laplacian matrix, we need to assume that  the transform matrix is perfectly known, so that
    $\hat{\mU}_{\mathcal{K}}=\mU_{\mathcal{K}}$. Even though these are ideal conditions because our proposed algorithms are only able to recover $\mU_{\mathcal{K}}$ up to a rotation, the closed forms provide a meaningful insight into the relation between graph sparsity and signal bandwidth. Furthermore, we remark that it is possible to recover the matrix  $\mU_{\mathcal{K}}$ with no rotation in the case where the signal values belong to a known discrete set, by using the blind algorithm of \cite{Macchi}, as shown in Section VI. Assume w.l.o.g. that the  $K$ eigenvectors of $\mU_{\mathcal{K}}$ are sorted
  in increasing order of their corresponding   eigenvalues so that $0=\lambda_{\mathcal{K},1}\leq\lambda_{\mathcal{K},2}\leq\ldots \leq \lambda_{\mathcal{K}, K}$ with $\mLambda_{\mathcal{K}}=\text{diag}(\boldsymbol{\lambda}_{\mathcal{K}})$.
  As a consequence, the  convex set   $\mathcal{X}(\mU_{\mathcal{K}})$ of problems $\mathcal{P}_{f_1}$ and $\mathcal{P}_{f_2}$, can be written as
\beq \label{eq_system}
\mathcal{X}(\mU_{\mathcal{K}})\triangleq \left\{\begin{array}{lllll} \text{a)} \quad \mL \mU_{\mathcal{K}}= \mU_{\mathcal{K}} \mLambda_{\mathcal{K}} \medskip\\
\text{b)} \quad \mL \mathbf{1}=\mathbf{0}\medskip \\
\text{c)} \quad L_{i j}=L_{j i}\leq 0,\, \forall i\neq j \medskip \\
\text{d)} \quad \lambda_{\mathcal{K},1}=0, \lambda_{\mathcal{K},2}\geq 0 \medskip \\
\text{e)} \quad   \lambda_{\mathcal{K},i+1}\geq \lambda_{\mathcal{K},
i}, \, i=2,\ldots, K-1 \medskip \\
\text{f)} \quad \text{tr}(\mL)=p
\end{array} \right.
\eeq
where conditions $\text{b)}-\text{e)}$   impose the desired Laplacian structure to $\mL$ and $\text{f)}$
fixes its $l_1$-norm  by avoiding the trivial solution.
For the sake of simplicity, we focus on the case where the subset  $\mathcal{K}$  consists
of the indices associated to the first $K$ columns of $\mU$; otherwise one can follow similar arguments for any subsets of indices
by properly rewriting conditions $\text{d)}-\text{e)}$.
As shown in Appendix \ref{B:closed_form_X},
we can reduce equations a)-c) and f) in (\ref{eq_system})
to the following compact form:
\beq \label{eq_matrix_system}
 \begin{array}{lllll} \mF \mx=\mb, \quad \quad \mx \in \mathbb{R}^{ N (N-1)/2+ K-1}_{+}
\end{array}
\eeq
where we defined $\mx\triangleq [-\mz ; \bar{\boldmath{\mlambda}}]$; $\mz \triangleq \text{vech}(\mL) \in \mathbb{R}^{N (N-1)/2}_{-}$; $\text{vech}(\mL)$ the half-vectorization of $\mL$ obtained
  by  vectorizing only the lower triangular part (except the diagonal entries)  of $\mL$;  $\mathbb{R}_{+}$ and $\mathbb{R}_{-}$
   the sets of, respectively, nonnegative and nonpositive real numbers;
 $\bar{\boldmath{\mlambda}}\triangleq \{\lambda_{\mathcal{K},i}\}_{\forall i \in \mathcal{K}^{-}}$
where, assuming the entries of $\boldmath{\mlambda}_{\mathcal{K}}$ in increasing order,  the index set $\mathcal{K}^{-}$ is obtained by removing from $\mathcal{K}$ the first index corresponding to $\lambda_{\mathcal{K},1}$; $\mb=[\mathbf{0}_{K N};\, p]$; and, finally, the coefficient matrix $\mF \in \mathbb{R}^{m\times n}$,
with $m=K N+1$, $n=N (N-1)/2+ K-1$ is defined in Appendix \ref{B:closed_form_X}, equation (\ref{eq:F_matrix}).

Note that the  rank $q$ of the coefficient matrix $\mF \in \mathbb{R}^{m\times n}$ is  $q \geq K-1$
and $q \leq \min \{n, m\}$.
\begin{proposition} \label{ranK_U}
\textit{Assume  the set $\mathcal{X}(\mU_{\mathcal{K}})$ to be feasible, then it holds:\\
a) $K-1 \leq \text{rank}(\mF)\leq K N+1$ for  $K \leq \frac{N}{2}-\frac{2}{N-1}$ and $K-1 \leq  \text{rank}(\mF)\leq N (N-1)/2+ K-1$  for $K > \frac{N}{2}-\frac{2}{N-1}$; \\ b) if $\text{rank}(\mF)= K N+1$
 and $K = \frac{N}{2}-\frac{2}{N-1}$, or $\text{rank}(\mF)=N (N-1)/2+ K-1$  and $K >\frac{N}{2}-\frac{2}{N-1}$, the set $\mathcal{X}(\mU_\mathcal{K})$
is a singleton.}
\end{proposition}
\begin{proof}
 See Section A in  Appendix C.
 \end{proof}\\
 By using some properties about  nonnegative solutions of underdetermined systems \cite{Wang},  under the assumption of existence and uniqueness of solutions for the system  (\ref{eq_matrix_system}),  we can derive the following condition relating graph sparsity and signal bandwidth $K$.
\begin{proposition} \label{ident_cond}
\textit{If the set $\mathcal{X}(\mU_{\mathcal{K}})$ with $K < \frac{N}{2}-\frac{2}{N-1}$, is feasible and   $\{\mx \, | \, \mF \mx=\mF \mx_0, \mx \geq \mathbf{0}\}$ is a singleton  for any nonnegative $s$-sparse signal $\mx_0$, then $\frac{N}{2}-\frac{2}{N-1} > K\geq 2 s/N$ and $\parallel\mA \parallel_0 \leq K (N-2) +2 c$ where $c$ is the number of connected components in the graph.}
\end{proposition}
\begin{proof}
 See Section B in Appendix C.
 \end{proof}\\
Note that if we assume  the  graph to be disconnected, i.e. $c>1$, the upper bound on the graph sparsity tends to increase to compensate for the connectivity loss.

 \begin{figure*}[t!]
    \begin{center}
        \begin{subfigure}[b]{0.28\textwidth}
        \hspace{-0.6in}
    \vspace{0.1in}
    \centering
        \includegraphics[width=2.2in, height=2.1in]{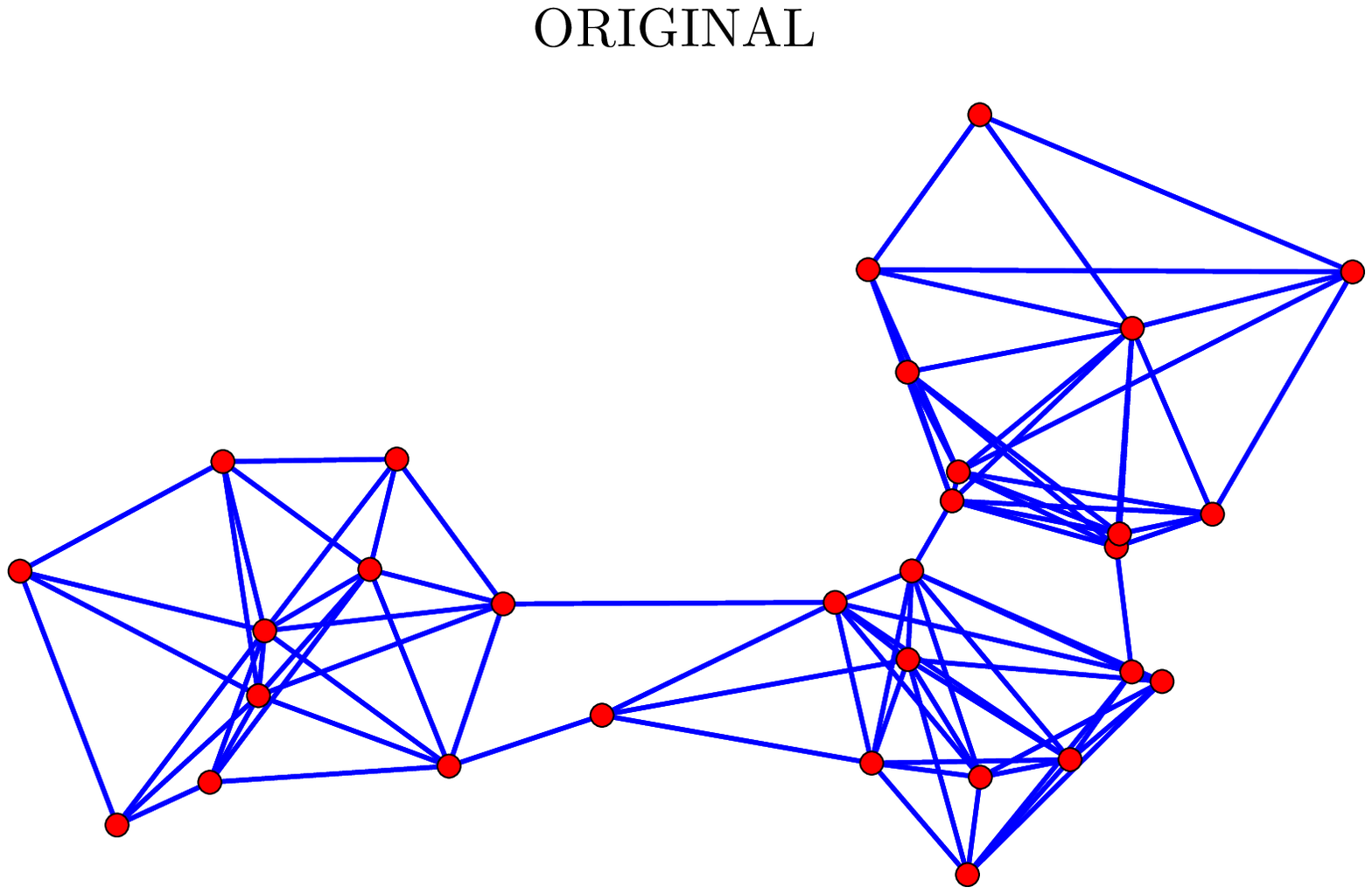}%
        \caption{\hspace{1.5cm}}\label{fig:graphs}
            \end{subfigure}%
            \hspace{-0.1in}
              \begin{subfigure}[b]{0.28\textwidth}
        \includegraphics[width=2.0in, height=2.3in, angle =-90]{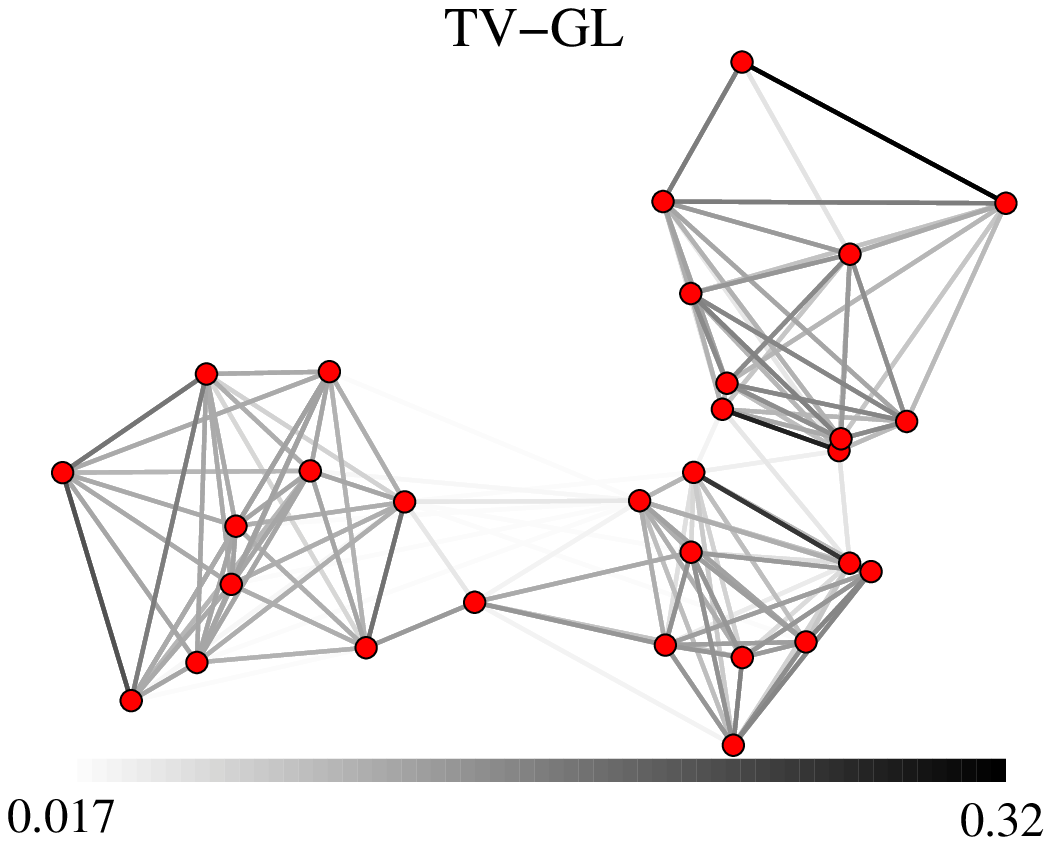}
          \caption{}\label{fig:graph_Alg1}
           \end{subfigure}%
           \qquad \qquad
            \begin{subfigure}[b]{0.28\textwidth}
        \includegraphics[width=2.0in, height=2.28in,  angle =-90]{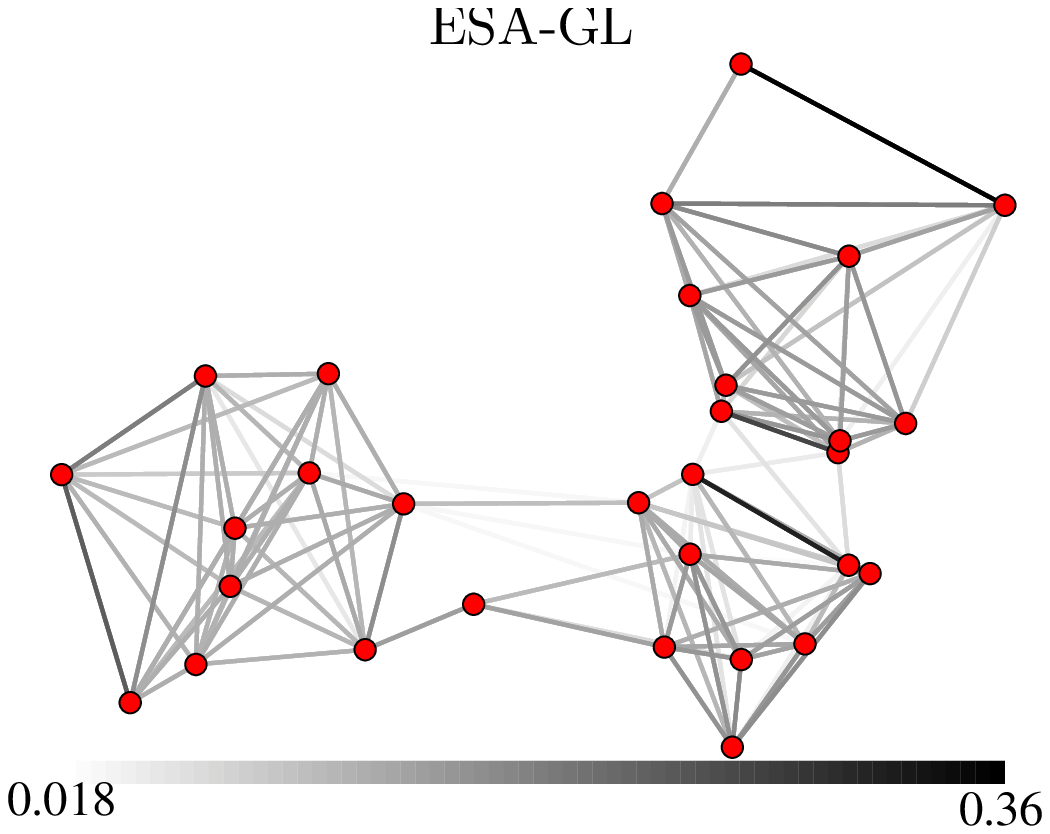}
        \caption{}\label{fig:graph_Alg2}
           \end{subfigure}%
    \caption{ Examples of graphs learning: $(\textit{a})$  original graph; $(\textit{b})$  recovered graph through TV-GL algorithm with $\mu=2$; $(\textit{c})$  recovered graph through ESA-GL algorithm $\mu=0.1$.}\label{fig:allgraphs}
    \end{center}
\end{figure*}

\vspace{-0.15cm}
\section{Numerical results}
In this section, we present some numerical results validating the effectiveness of the proposed graph-learning strategies
for  both synthetic and real-world graphs. In all numerical experiments we solved the proposed optimization problems by using
the convex optimization package CVX \cite{cvx}. \\

\noindent\textbf{\textit{Performance over synthetic data}}

Let us consider the graph in  Fig. \ref{fig:graphs} composed of $N=30$ nodes forming $3$ clusters of $10$ nodes each.  To start testing the effectiveness of the proposed graph inference strategies, in Fig. \ref{fig:allgraphs} we illustrate an example of graph topology recovery resulting by  using the  TV-GL  and ESA-GL algorithms.  
The color on each edge encodes the strength of that edge. We can observe that, learning the Fourier basis using Algorithm $1$,  both algorithms are able to ensure a good topology recovery.

As a statistical test, we run simulation over $100$ independent signal matrix realizations, with $N=30$ and $M=15$, assuming $K=3$ for the block sparsity and setting $\text{tr}(\mL)=N$. As a first performance measure, we measured the correlation coefficient between the true Laplacian entries $L_{ij}$ and their estimates $\hat{L}_{ij}$: ${\rho}(\mL,\hat{\mL})\triangleq \frac{\sum_{ij}(L_{ij}-L_{m})(\hat{L}_{ij}-\hat{L}_{m})}{\sqrt{\sum_{ij}(L_{ij}-L_{m})^2} \sqrt{\sum_{ij}(\hat{L}_{ij}-\hat{L}_{m})^2}}$ where $L_m$ and $\hat{L}_m$ are the  average values of the entries, respectively, of the true and estimated Laplacian matrices. Note that, because of  the Laplacian matrix structure, we always have $L_m=\hat{L}_m=0$.
In Fig. \ref{fig:rho_mu1},   we plot the average correlation coefficient $\bar{\rho}(\mL,\hat{\mL})$ vs. the penalty coefficient $\mu$, where $\hat{\mL}$ is computed by applying the total variation based (TV-GL) or the estimated-signal aiding  graph-learning (ESA-GL) algorithms.\\
  To get insight into the proposed algorithms, we considered both cases  when the Fourier transform matrix $\mU$ is a priori known or when it is estimated
by using Algorithm $1$. From Fig.  \ref{fig:rho_mu1},
we notice that both methods are able to ensure  high  average correlation coefficients.
   Furthermore, we can observe a high robustness of the TV-GL method to the choice of the penalty parameter with respect to the ESA-GL algorithm.
   By comparing the curves obtained by assuming perfect knowledge of $\mU$ with those derived by  estimating it  through Algorithm $1$, we can also notice that the performance loss due to estimation errors is negligible.\\
In general, the choice of $\mu$ has an impact on the final result. To assess this impact, in Fig. \ref{fig:perc_link} we illustrate the average normalized recovery error  $\bar{\mathcal{E}}_0$ versus the penalty coefficient $\mu$.
The error $\bar{\mathcal{E}}_0$ represents the fraction of  misidentified  edges and is defined as $\frac{\parallel \mathbf{A}-\hat{\mathbf{A}}_b\parallel_0}{N\cdot(N-1)}$ where $\mathbf{A}$ and $\hat{\mathbf{A}}_b$ are, respectively, the groundtruth and the recovered binary adjacency  matrices.
 The binary matrix $\hat{\mathbf{A}}_b$ has been obtained by thresholding the entries of the recovered matrix $\hat{\mathbf{A}}$ with a threshold equal to half the average values of the elements of $\hat{\mathbf{A}}$.
 We consider both cases where the Fourier basis vectors are estimated (upper plot) or   a-priori perfectly known (lower plot).
 Under this last setting, we solve problem $\mathcal{P}_{f_1}$ and $\mathcal{P}_{f_2}$ by assuming $(\mL,\mC_{\mathcal{K}}) \in \mathcal{X}(\mU_{\mathcal{K}})$. Thereby, we solve problem $\mathcal{S}_k$,  by supposing  perfect knowledge of $\mU$, in order to calculate
 the estimated signals $\hat{\mS}$,
needed as input to solve problem $\mathcal{P}_{f_2}$.
  We can see that the percentage of incorrect edge recovery can be in the order of $0.4\%$.
 Furthermore, comparing the TV-GL algorithm with the ESA-GL, we can observe that TV-GL tends to perform better.

\noindent\textbf{\textit{Performance versus signal bandwidth}}\\
 To evaluate the impact of the  signal bandwidth $K$ on the graph recovery strategy, we illustrate some numerical results performed on random graphs composed
 of $N=30$ nodes with $3$ clusters, each of $10$ nodes, by averaging the final results over $100$ graphs and signal matrix realizations for  $M=15$.
 We consider both cases of exactly or only approximately
(or compressible)  bandlimited graph signals \cite{Foucart}.
We recall that a vector is compressible if the error of its best $k$-term approximation decays quickly in $k$ \cite{Foucart}.
 We generate each band-limited signal $\bs_i$ for $i=1,\ldots, M$   as  $s_i(k)\sim \mathcal{N}(1,0.5)$ for all $k\leq  K$, and
   $s_i(k)=0$ for all $k>K$.
 In Fig. \ref{fig:bandlimited} we plot the average  recovery error  $\bar{\mathcal{E}}_0$ (upper plot)
  and the average recovery error $\bar{\mathcal{E}}_F$ (lower plot), defined as
$\frac{\parallel \mathbf{A}-\hat{\mathbf{A}} \parallel_F}{N\cdot(N-1)}$,  vs. the signal bandwidth $K$.
We selected for each $K$ the optimal coefficient $\mu$ minimizing the average recovery error.
 We can observe that the error tends  to increase as  the signal bandwidth
 $K$  gets larger than the number of clusters, equal to $K=3$.
 \begin{figure}
\centering
\includegraphics[width=8.6cm,height=5.2cm]{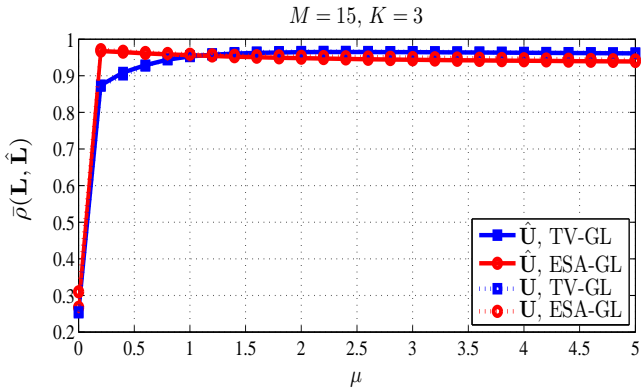}
\caption{Average correlation coefficient  versus the parameter $\mu$.}
\label{fig:rho_mu1}
\end{figure}
 \begin{figure}[h]
\centering
\includegraphics[width=8.5cm,height=7.6cm]{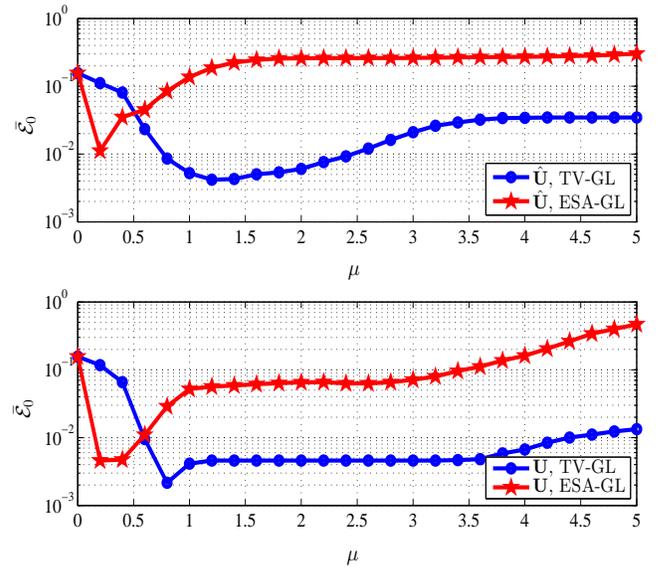}
\caption{Average recovery error  vs. the  coefficient $\mu$, using  estimated Fourier basis vectors (upper plot) and  a-priori knowledge of the Fourier transform (lower plot).}
\label{fig:perc_link}
\end{figure}
Finally,  in  Fig. \ref{fig:compressed} we report the averaged recovery errors in case of
 compressible graph signals,  generated as $s_i(k)\sim \mathcal{N}(1,0.5)$ for all $k\leq  K_v$
and   $s_i(k)=(K_v/k)^{2 \beta}$  for all $k>K_v$, with $\beta=2$, $K_v=5$ \cite{Kovacevic}.
We can observe  that for the TV-GL method the average recovery errors  increases as the signal bandwidth increases, while the
minimum is achieved for $K=5$, since $K_v=5$  represents  the approximated  graph signal bandwidth.
Additionally, the ESA-GL method seems to reach minimum values quite similar to those of  the TV-GL algorithm for $K$  close to $5$.\\

\begin{figure*}
\centering
\begin{subfigure}[b]{0.45\textwidth}
\centering
               \includegraphics[width=\textwidth]{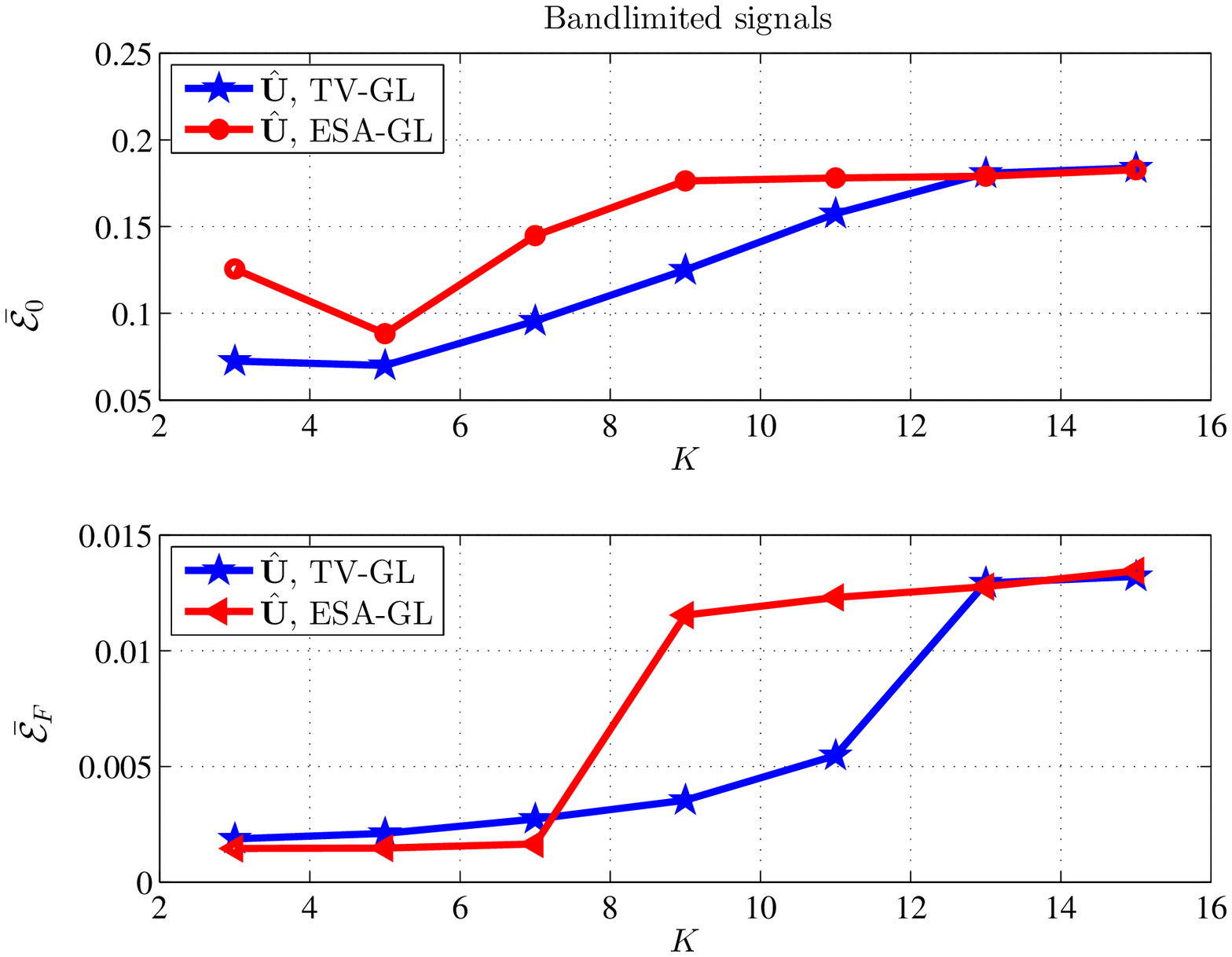}
               \caption{}\label{fig:bandlimited}
     \end{subfigure}
\hfill
\begin{subfigure}[b]{0.45\textwidth}
\centering
              \includegraphics[width=\textwidth]{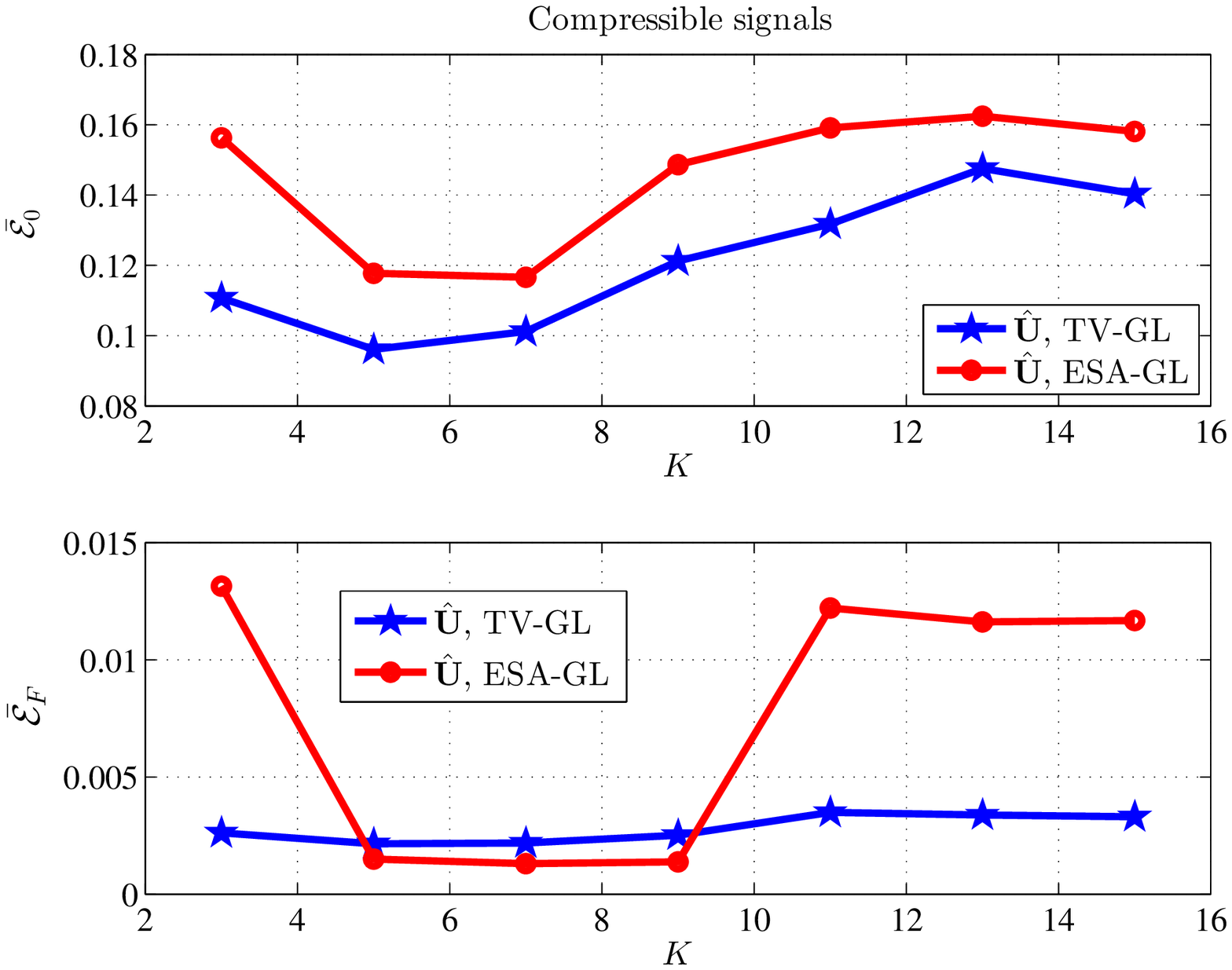}
              \caption{} \label{fig:compressed}
\end{subfigure}
   \caption{ Average recovery errors  versus $K$ for  a) bandlimited  graph signals; b) for  compressible  graph signals.}
\end{figure*}



\noindent\textbf{\textit{Performance over real data}}\\
In this section we test our methods over real data for recovering the brain functional connectivity network associated to epilepsy.
Understanding how this abnormal
neuronal activity emerges and from what epileptogenic zone, would help to refine surgical
 techniques and devise novel therapies.
The diagnosis involves comparing ECoG time series obtained from the patient's brain before and after onset of a seizure.\\

  \noindent \textit{Epilepsy data description}\\
We used the datasets taken from
experiments conducted in an epilepsy study \cite{Kramer} to infer the brain functional activity map. The data were collected from a 39-year-old female subject with a case of intractable epilepsy at the  University of California, San
 Francisco (UCSF) Epilepsy Center \cite{Kramer}. An $8 \times 8$ electrode grid was implanted in the
 cortical surface of the subject's brain  and two accompanying strips of six electrodes each were implanted deeper into the brain.
 This combined electrodes network recorded $76$ ECoG time series, consisting of voltage levels in the vicinity of each electrode,
  which are indicative of the levels of local brain activity. Physicians recorded data for $5$ consecutive days and ECoG epochs containing eight seizures were
  extracted from the record. The time series at each electrode were first passed through a bandpass filter with cut-off frequencies of $1$
and $50$ Hz and two temporal intervals of interest were picked for analysis, namely, the preictal and ictal  intervals.
 The preictal interval is defined as a $10$-second  interval preceding seizure onset, while the ictal interval corresponds to the  $10$-second immediately after the start of a seizure. For further details on data acquisition see  \cite{Kramer}. \\

\noindent \textit{Recovery of brain functional activity graph}\\
Since the observed signal is highly non-stationary, following the same approach described in  \cite{Kramer}, before running our algorithms we divide the $10$ seconds interval into $20$ overlapping segments of $1$ second, so that each segment overlaps with  the previous one by $0.5$ seconds. The reason is that within a one second interval, the brain activity is approximately stationary.
After this segmentation, our goal is to infer the network topology for each segment.
We denote by $\mY \in \mathbb{R}^{N \times M}$  the observed data matrix,  with $N=76$ and $M=400$. Before applying our inference algorithm, we filter the data to reduce the effect of noise, using the method proposed in \cite{Donoho} where an optimal shrinkage of the empirical singular values of the observed matrix is applied.
Given the compressed signal matrix $\hat{\bY}$,  we  proceed by running first Algorithm $1$ to estimate the transform matrix $\hat{\mU}$
and, thereafter, we recover the brain network topology by using the TV-GL  algorithm.
In Fig. \ref{fig:pre_ict_int}  we show two examples of graphs learned from the observed (and filtered) data. They refer, respectively, to the pre-ictal interval $19$ and  ictal interval $1$.
\begin{figure*}
\centering
\begin{subfigure}[b]{0.32\textwidth}
\centering
               \includegraphics[width=\textwidth]{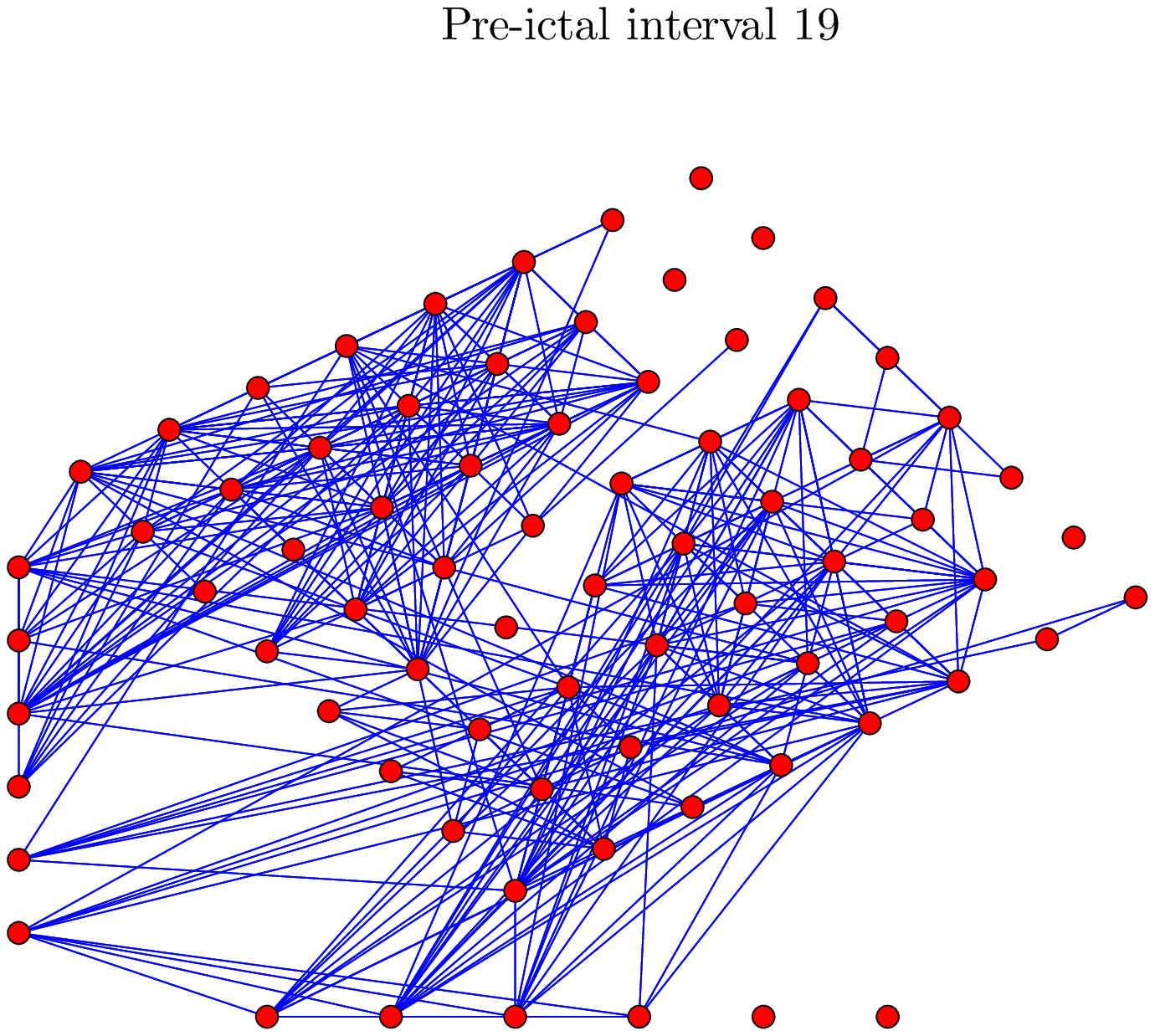}
     \end{subfigure}
     \hspace{1.4cm}
\begin{subfigure}[b]{0.32\textwidth}
\centering
              \includegraphics[width=\textwidth]{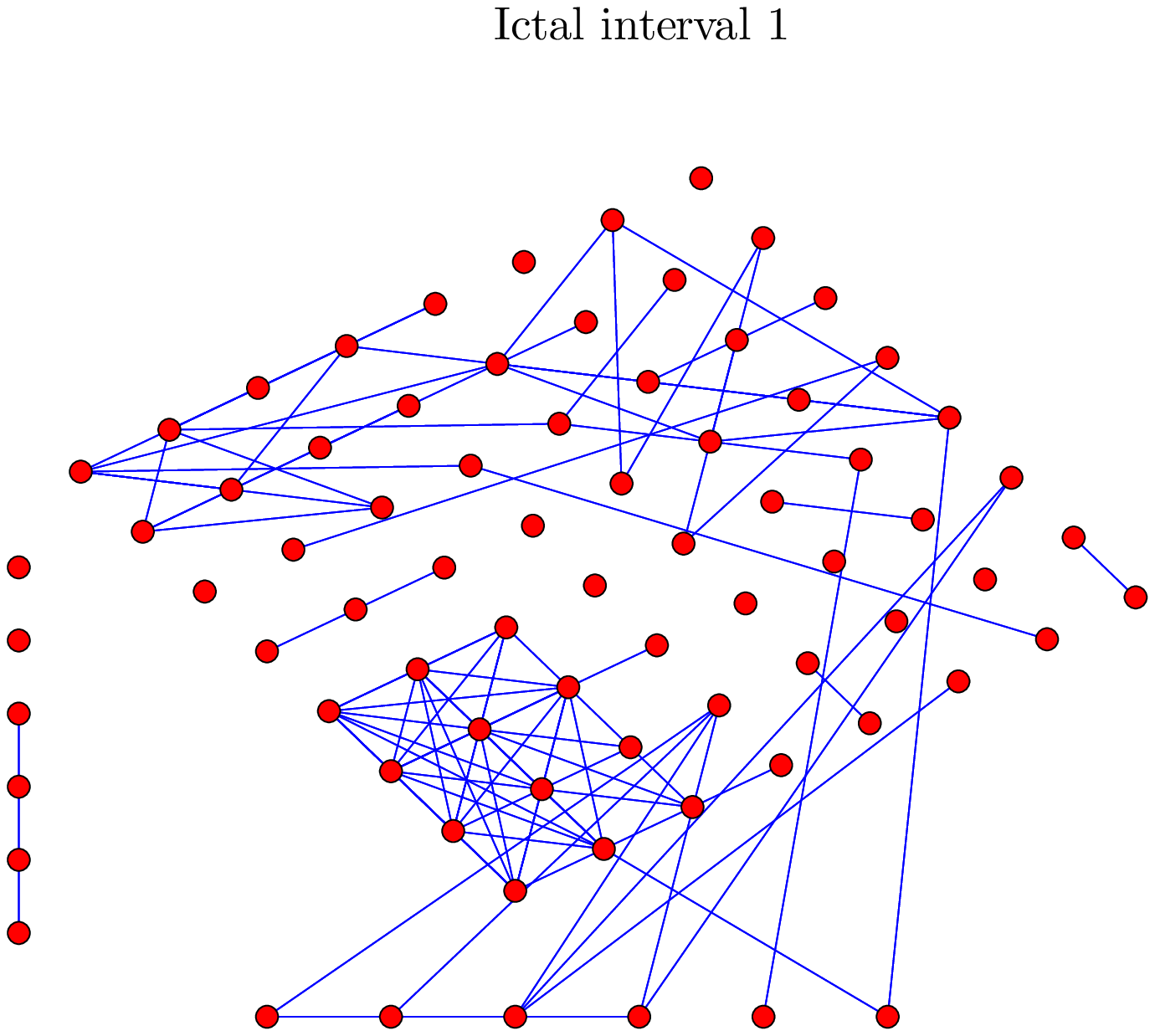}
\end{subfigure}
       \vskip \baselineskip
\caption{Recovered  networks for the pre-ictal interval 19 (left) and the first  ictal interval (right).}\label{fig:pre_ict_int}
\end{figure*}
We can notice how the number of edges changes dramatically, showing that the network connectivity tends to decrease at seizure onset.
This behavior is especially evident in the bottom corner of the grid closest to the two strips of electrodes which were located close to where the seizure was suspected to emanate so that in this region the connectivity of the network tends to decrease at seizure onset, as observed in \cite{Kramer}.


The problem in associating a graph topology to a set of signals is that we do not know the ground truth. To validate our approach from a purely data-oriented perspective, we used the following approach. We consider the observations taken in two intervals, say pre-ictal interval $19$ and ictal interval $1$. The graphs inferred in these two cases are the ones depicted in Fig. \ref{fig:pre_ict_int}. Then, we assume that we observe only the signals taken from a subset of electrodes and we ask ourselves whether we can recover the signals over the unobserved electrodes. If the signals is band-limited over the recovered graph, we can apply sampling theory for graph signals, as e.g. \cite{Tsit_Barb_PDL}, to recover the unobserved signals. Then we compare what we are able to reconstruct with what is indeed known. We use a greedy sampling strategy using an E-optimal design criterion \cite{di2017sampling}, selecting the bandwidth provided by our transform learning method, which is equal to $60$ for the pre-ictal interval, and to $64$ for the ictal interval. The number of electrodes assumed to be observed is chosen equal to the bandwidth in both cases. For each time instant, we use a batch consistent recovery method that reconstructs the signals from the collected samples [see eq. (1.9) in \cite{di2017sampling}]. In Fig. \ref{fig:rec_brain}, we illustrate the true ECoG signal present at node $72$ (black line) and what we are able to reconstruct (green line) from a subset of nodes that does not contain node $72$. We repeated this operation for both the pre-ictal (top) and the ictal (bottom) phases. As we can notice from Fig. \ref{fig:rec_brain}, the reconstructed signal is very close to the real signal, in both phases. This means that, evidently, the inferred graphs are close to reality, because otherwise the sampling theory would have been based on a wrong topology assumption and then it could not lead to such a good prediction.\\
\begin{figure}[h]
\centering
\includegraphics[width=7cm,height=6.5cm]{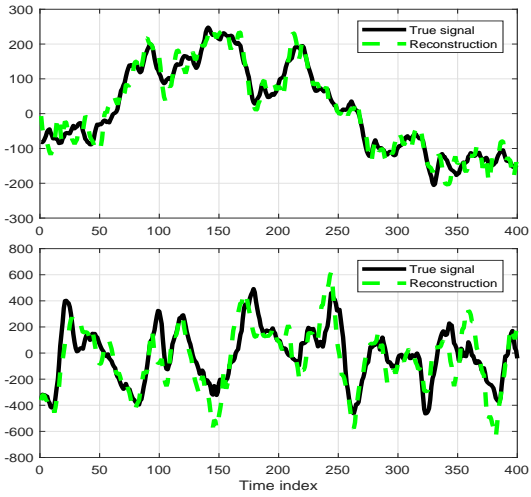}
\caption{True and recovered brain signals during pre-ictal (top) and ictal (bottom) intervals.}\label{fig:rec_brain}
\end{figure}

\textbf{\textit{Comparison with GSP-based topology inference methods.}}
In this section we compare the performance of our algorithms  with  recent GSP-based algorithms, namely
the methods used in \cite{Segarrajournal2016}, \cite{Kalofolias}, \cite{Frossard_jour}  to identify the Laplacian matrix exploiting the   smoothness of graph signals.
To make a fair comparison, we considered  for all methods the combinatorial Laplacian matrix.
To implement the algorithm
SpecTemp proposed in   \cite{Segarrajournal2016}, we need to solve the following problem
\beq \nonumber
\begin{array}{ll}
\underset{\mL,\overline{\mL}_{\bar{\mathcal{K}}}, \mL^{\prime}, \boldsymbol{\lambda}}{\min} \quad  \parallel \mL \parallel_{1} \hspace{1cm}(\mathcal{P}_{s})\\
\hspace{0.5cm} \begin{array}{llll} \mbox{s.t.}
      &\hspace{0.6cm}\mL \in \mathcal{L} \medskip\\
       &\hspace{0.6cm}\mL^{\prime} = \overline{\mL}_{\bar{\mathcal{K}}}+ \sum_{k=1}^{K} \lambda_k \hat{\mv}_k \hat{\mv}_k^T, \; \overline{\mL}_{\bar{\mathcal{K}}} \hat{\mV}_K=\mathbf{0}\medskip\\
       &\hspace{0.6cm} d(\mL^{\prime},\mL)\leq \epsilon
      \end{array}
\end{array}
\eeq
where $\hat{\mV}_K\triangleq [\hat{\mv}_1, \ldots, \hat{\mv}_K]$ is the
graph Fourier basis obtained via PCA, i.e. taking the
$K$ eigenvectors associated to the $K$ largest nonzero eigenvalues  of the sample covariance matrix $\mC_y \triangleq \frac{1}{M} \mY \mY^T$, with $\mY$ assumed to be zero mean; $d(\cdot,\cdot)$ is a convex vector distance function, e.g. $\parallel \mL - \mL^{\prime}\parallel_F$, and $\epsilon$
    is a positive constant controlling the feasibility of the set.
 We select $\epsilon$ as the smallest value that admits a feasible solution.
As suggested in \cite{Segarrajournal2016}, the recovered eigenvalues $\boldsymbol{\lambda}$ are required to satisfy the condition $\lambda_i\geq \lambda_{i+k}+\delta$, $\forall i$, with $k=2$ and $\delta=0.1$. This is done to enforce the property that the principal eigenvectors of the estimated covariance matrix give rise to the eigenvectors associated to the smallest eigenvectors of the Laplacian matrix.\\
We also considered  the combinatorial Laplacian recovering algorithm proposed by Dong et al. in \cite{Frossard_jour}, which solves the following optimization problem
\beq \nonumber
\begin{array}{ll}
\underset{\mL}{\min} \quad \mu \parallel \mL \parallel_{F}^{2}+ \text{tr}( \mY^T \mL \mY) \hspace{1cm}(\mathcal{P}_{D})\\
 \begin{array}{llll} \mbox{s.t.}
      &\hspace{0.2cm}\mL \in \mathcal{L}, \; \text{tr}(\mL)=p
      \end{array}
\end{array}
\eeq
where the coefficient $\mu>0$ is chosen in order to maximize the correlation coefficient $\rho(\mL,\hat{\mL})$.
Finally, we considered
the Kalofolias algorithm proposed in  \cite{Kalofolias}, which finds the adjacency matrix as the optimal solution of the following convex problem
\beq \nonumber
\begin{array}{ll}
\underset{\mA \in \mathcal{A}_N}{\min} \quad  \text{tr}( \mA \mZ)-\alpha \mathbf{1}^T \log(\mA \mathbf{1})+\beta
\parallel  \mA\parallel_F^{2} \hspace{1cm}(\mathcal{P}_{K})
\end{array}
\eeq
where $\alpha, \beta>0$, $\mathcal{A}_N \triangleq \{\mA \in \mathbb{R}_{+}^{N \times N} \; : \; \mA=\mA^T, \text{diag}(\mA)=0, A_{ij}>0\}$ and $\mZ\triangleq (Z_{i j})_{\forall i,j}$
 with $Z_{ij}\triangleq \| \my_i-\my_j\|^2$.
Note that the logarithmic barrier in $\mathcal{P}_{K}$ forces the node degrees to be positive, while the Frobenius norm is added to control sparsity. \\
As figures of merit, we considered the average correlation coefficient $\bar{\rho}(\mL,\hat{\mL})$, the average recovery error $\bar{\mathcal{E}}_0$ and three evaluation criteria commonly used in information retrieval \cite{Manning}, i.e. the F-measure, the edge \mbox{Precision} and the edge \mbox{Recall}. The \mbox{Precision} is the percentage of correct edges in the learned graph and the  \mbox{Recall} evaluates the percentage of the edges in the true graph that are present in the learned graph \cite{Manning}.
Defining $\mathcal{E}_{g}$, $\mathcal{E}_{r}$  as the sets  of edges present, respectively, in the ground truth  and in the recovered graph, Precision and Recall are
\beq
\mbox{Precision}= \frac{\mid \mathcal{E}_{g} \cap \mathcal{E}_{r} \mid}{\mid \mathcal{E}_{r}\mid },\quad\quad
\mbox{Recall}= \frac{\mid \mathcal{E}_{g} \cap \mathcal{E}_{r} \mid}{\mid \mathcal{E}_{g}\mid}.
\eeq
The F-measure is defined as the
harmonic mean of edge \mbox{Precision} and edge \mbox{Recall}.
In our numerical experiments, we average the results
over $100$ random graphs and compare the different methods of interest on three different types of graphs composed of $N=30$ nodes and  for $M=100$ graph signal vectors.
More specifically, we considered three cases: i) random graphs composed of three  clusters, each of $10$ nodes, with edge probability of connection among
clusters equal to $1.e-2$;
ii)   Erd\H{o}s-R\`{e}nyi graphs with probability of connection $0.3$; iii) Bar\'{a}basi-Albert graphs generated  from a set of $m_0=4$ initial nodes and where one node is added at each step and connected to $m=3$  already existing nodes according to the
preferential attachment mechanism. For a fair comparison we generated the signals according to two models. The graph signals are generated as $\by=\mU \bs$,  where  $\bs$ are: i) $\mathcal{K}$-bandlimited random i.i.d. signals,  with zero mean, unit variance and $|\mathcal{K}|=3$ (bandlimted GS model); ii)
zero-mean Gaussian random variables with covariance matrix defined as the pseudo-inverse of the eigenvalue matrix $\mathbf{\Lambda}$ of the graph Laplacian, i.e. $\bs \sim \mathcal{N}(\mathbf{0},\mathbf{\Lambda}^{\dag})$ \cite{Frossard_jour} (inverse Laplacian GS model).
To compute the Precision, Recall and F-measure, we applied a threshold on the estimated edge values, to identify the sets $\mathcal{E}_{g}$ and $\mathcal{E}_{r}$. We used the same threshold for all methods, evaluated as one half the average value of the off-diagonal entries of the estimated Laplacian.
In Table \ref{tab:comperison},  we summarize  the performance measures, respectively, for the bandlimited and the inverse Laplacian graph signal models. To make a fair comparison, we put every method in its best conditions. So,
for SpecTemp \cite{Segarrajournal2016}, we searched for the minimum $\epsilon$  value guaranteeing the set feasibility, while for the Kalofolias \cite{Kalofolias} and Dong et al. \cite{Frossard_jour} algorithms, we selected the optimal $\mu$, $\alpha$ and $\beta$ coefficients achieving  the maximum $\rho(\mathbf{L},\bar{\mathbf{L}})$. Comparing all methods, we notice first that our methods offer superior performance when we observe $\mathcal{K}$-bandlimited signals over random graphs composed of $K$ clusters. Clearly this had to be expected because our methods are perfectly matched to such a condition. However, we can also notice that our methods are quite robust across the different graph and signal models. Additionally, it can be noted that the SpecTemp  algorithm performs poorly when we use a low  number of observed signals $M$ and of  basis vectors $K$,  whereas the TV-GL and ESA-GL methods reach good performance under this critical setting. In Table \ref{tab:comperison2} we report the performance metrics for different random geometric graphs with $K=3,6,9$ clusters, each composed of $10$ nodes and with probability of connection among clusters $1e-2$. Note  that the proposed graph learning methods ensure good  performance for both bandlimited  and inverse Laplacian GS models.
Finally, in Table \ref{tab:comperison3}, under the setting of discrete   alphabet of size $K$, we report the numerical results by using as input of  the TV-GL and ESA-GL methods the de-rotate transform matrix $\hat{\mU}_{\mathcal{K}}\mH^T$ where the matrix $\mH$  is estimated  according to the blind recovering method proposed in \cite{Macchi}.  We named these methods Derotated TV-GL (DTV-GL) and
Derotated ESA-GL (DESA-GL). Interestingly, we can note  that high performance levels are achieved  in both cases, by applying the derotation method  or by omitting  this. Indeed, we observed that
for bandlimited graph signals over random graphs with clusters, as the number of clusters increases, the transform and signals matrices, learned through the TV-GL and ESA-GL methods, are more and more accurate and are not rotated    with respect to the true transform matrix. This is indeed an interesting behavior whose investigation we defer to future works.

\begin{table*}[t]
 \hspace{4.5cm}{\bf{Bandlimited GS model}} \hspace{4.5cm}  \hspace{0.8cm}
 {\bf{Inverse Laplacian GS model}} \newline

 \centering
  \begin{tabular}{l l l l l l l l l l l l}
    \cline{1-6}
      \cline{8-12}

     & \bf{TV-GL} & \bf{ESA-GL} & \bf{SpecTemp} &  \bf{Kalofolias}  & \bf{Dong et al.} & & \bf{TV-GL} & \bf{ESA-GL} & \bf{SpecTemp} &  \bf{Kalofolias}  & \bf{Dong et al.}\\ \hline

     \bf{RG, $K=3$ clusters }  &  &  &  & & &  & & & & & \\

   F-measure  & 0.792 & 0.839 & 0.311 & 0.774 & 0.768  & & 0.688 & 0.650& 0.318 & 0.879 & 0.865\\

  Precision & 0.664 & 0.734 & 0.185  & 0.734  & 0.650 & & 0.540& 0.498 & 0.189 & 0.901  & 0.794\\

  Recall & 0.989 & 0.982 & 0.980  & 0.829   & 0.943 & & 0.955 & 0.946 & 1  & 0.859   & 0.952  \\

  $\bar{\mathcal{\rho}}(\mathbf{L},\hat{\mathbf{L}})$  &  0.947 &  0.953 & 0.886  &  0.905 & 0.934& &  0.960 &  0.937 & 0.892 &  0.955 & 0.974 \\

  $\bar{\mathcal{E}}_0$  & 0.097 & 0.071 & 0.801 & 0.090 & 0.107 & & 0.164 & 0.194 & 0.810 & 0.044 & 0.056\\


  \bf{Erd\H{o}s-R\`{e}nyi}  &  &  &  &  & & & & & &\\

   F-measure  & 0.433 &0.480   & 0.237& 0.444 & 0.401 & & 0.545  & 0.448   & 0.237 & 0.633 & 0.712 \\

  Precision & 0.280 &  0.323 &  0.134 &  0.331 & 0.255 & & 0.383 &  0.308 &   0.134 & 0.501 & 0.577\\

  Recall & 0.961& 0.943  &  1 & 0.700 & 0.954 & & 0.946& 0.826 &  1 & 0.863 & 0.932\\
   $\bar{\mathcal{\rho}}(\mathbf{L},\hat{\mathbf{L}})$  & 0.889  & 0.896  & 0.844  &  0.848 & 0.879 & & 0.949 &  0.894 & 0.844 & 0.932  & 0.957\\

  $\bar{\mathcal{E}}_0$  & 0.339 &  0.275&0.865 & 0.243 & 0.392& &0.212 & 0.273 & 0.865 & 0.136 & 0.102\\


  \bf{Barab\'{a}si-Albert}  &  &  &  & &  & & & & & &\\

       F-measure  &  0.436    & 0.462 & 0.301 & 0.464 &  0.448  &  & 0.549 & 0.433& 0.304 & 0.652 & 0.670\\

         Precision & 0.287  &  0.313& 0.178 &  0.364 &  0.294 & & 0.440& 0.311& 0.179 &  0.579 & 0.662\\
         Recall & 0.914 &0.886  &   0.990 & 0.653 & 0.936 &  & 0.731& 0.718&  1  & 0.750 &  0.680\\
           $\bar{\mathcal{\rho}}(\mathbf{L},\hat{\mathbf{L}})$  &  0.792 & 0.795  & 0.768 &  0.784 &0.790  &  &  0.920& 0.798& 0.773 &  0.924 & 0.915\\

  $\bar{\mathcal{E}}_0$  & 0.423 &0.370  & 0.815 & 0.272  & 0.415 & & 0.215 & 0.336 &0.820 & 0.143  & 0.120\\

  \hline
    \end{tabular}
 \caption{Performance comparison between  TV-GL, ESA-GL, SpecTemp \cite{Segarrajournal2016}, Kalofolias \cite{Kalofolias} and Dong et al. \cite{Frossard_jour} for different graph signal models, with $M=100$. For TV-GL and ESA-GL we set $K=3$.}\label{tab:comperison}
\end{table*}

\begin{table*}[t]
 \hspace{4.5cm}{\bf{Bandlimited GS model}} \hspace{4.5cm}  \hspace{0.8cm}
 {\bf{Inverse Laplacian GS model}} \newline

 \centering
  \begin{tabular}{l l l l l l l l l l l l}
    \cline{1-6}
      \cline{8-12}

     & \bf{TV-GL} & \bf{ESA-GL} & \bf{SpecTemp} &  \bf{Kalofolias}  & \bf{Dong et al.} & & \bf{TV-GL} & \bf{ESA-GL} & \bf{SpecTemp} &  \bf{Kalofolias}  & \bf{Dong et al.}\\ \hline

     \bf{RG, $K=3$ clusters }  &  &  &  & & &  & & & & & \\

   F-measure  & 0.792 & 0.839 & 0.311 & 0.774 & 0.768  & & 0.688 & 0.650& 0.318 & 0.879 & 0.865\\

  Precision & 0.664 & 0.734 & 0.185  & 0.734  & 0.650 & & 0.540& 0.498 & 0.189 & 0.901  & 0.794\\

  Recall & 0.989 & 0.982 & 0.980  & 0.829   & 0.943 & & 0.955 & 0.946 & 1  & 0.859   & 0.952  \\

  $\bar{\mathcal{\rho}}(\mathbf{L},\hat{\mathbf{L}})$  &  0.947 &  0.953 & 0.886  &  0.905 & 0.934& &  0.960 &  0.937 & 0.892 &  0.955 & 0.974 \\

  $\bar{\mathcal{E}}_0$  & 0.097 & 0.071 & 0.801 & 0.090 & 0.107 & & 0.164 & 0.194 & 0.810 & 0.044 & 0.056\\

  \bf{RG, $K=6$ clusters}  &  &  &  &  & & & & & &\\

   F-measure  & 0.712  & 0.776   &0.177 & 0.752 &  0.713& & 0.661  &0.705    & 0.175 & 0.742 & 0.762  \\

  Precision & 0.557 &  0.639 &  0.097 & 0.671  & 0.586 & & 0.500 & 0.556  & 0.097   & 0.628 & 0.659 \\

  Recall &0.993 & 0.992  & 1 & 0.863 &0.917 & & 0.981& 0.967 & 0.980  & 0.913 & 0.907 \\
   $\bar{\mathcal{\rho}}(\mathbf{L},\hat{\mathbf{L}})$  & 0.953  &  0.962 &   0.886&  0.927 & 0.931 & & 0.951 & 0.955 & 0.873 & 0.939  & 0.937 \\

  $\bar{\mathcal{E}}_0$  & 0.078 & 0.056 & 0.902& 0.055 & 0.072& & 0.098& 0.079 & 0.886 & 0.062 & 0.055\\

  \bf{RG, $K=9$ clusters }  &  &  &  & &  & & & & & &\\

   F-measure  &  0.743&  0.875 & 0.087& 0.691 & 0.747 & & 0.617  & 0.688   &0.037  &0.750  & 0.749 \\

  Precision & 0.593 & 0.781  & 0.049  & 0.566  & 0.658 & & 0.451 & 0.536  & 0.139   &0.664  & 0.660\\

  Recall & 0.997& 0.999  & 0.376 & 0.898 & 0.867& & 0.980& 0.965 & 0.120  & 0.862 & 0.867\\
   $\bar{\mathcal{\rho}}(\mathbf{L},\hat{\mathbf{L}})$  &  0.966 & 0.980  & 0.785  & 0.942  & 0.939 & & 0.957 & 0.962 & 0.462 & 0.942  & 0.939\\

  $\bar{\mathcal{E}}_0$  & 0.047 & 0.019 & 0.532& 0.055& 0.039& & 0.083&  0.060& 0.320 & 0.039 & 0.039\\

  \hline
    \end{tabular}
 \caption{Performance comparison between  TV-GL, ESA-GL, SpecTemp \cite{Segarrajournal2016}, Kalofolias \cite{Kalofolias} and Dong et al. \cite{Frossard_jour} for random geometric graph with   $K=3,6,9$ clusters  and  $M=100$.}\label{tab:comperison2}
\end{table*}

\begin{table*}[t]
 \hspace{6.5cm}{\bf{Bandlimited, GS with discrete alphabet}} \newline

 \centering
  \begin{tabular}{l l l l l l l l }
    \cline{1-8}

     & \bf{TV-GL} & \bf{ESA-GL} & \bf{DTV-GL} & \bf{DESA-GL} & \bf{SpecTemp} &  \bf{Kalofolias} & \bf{Dong et al.} \\ \hline

     \bf{RG, $K=2$, $M=120$}  &  &  &  &  &  &  &\\

   F-measure  &  0.795& 0.827 & 0.821 & 0.811 & 0.449 & 0.711 & 0.762\\

  Precision & 0.669 &  0.726& 0.715 & 0.693 & 0.291 & 0.560 & 0.633\\

  Recall & 0.987 &0.966  & 0.971 &  0.983& 1 & 0.995 & 0.968\\

  $\bar{\mathcal{\rho}}(\mathbf{L},\hat{\mathbf{L}})$  & 0.945 & 0.947 & 0.943 & 0.948 & 0.902 & 0.927 &0.934\\

  $\bar{\mathcal{E}}_0$  &  0.146&  0.116& 0.121 & 0.131 & 0.708 & 0.237 & 0.174\\







  \bf{RG, $K=6$, $M=1000$ }  &  &  &  & &  & & \\

   F-measure  & 0.809  & 0.840  & 0.779&  0.839& 0.151 & 0.745 & 0.746 \\

  Precision &  0.685&  0.728 & 0.641  & 0.726  & 0.090 & 0.634& 0.634  \\

  Recall & 0.992& 0.997  & 0.997 & 0.997 &0.479 & 0.909& 0.913 \\
   $\bar{\mathcal{\rho}}(\mathbf{L},\hat{\mathbf{L}})$  &  0.957 & 0.970  & 0.959  &  0.968 & 0.804 & 0.939&   0.936 \\

  $\bar{\mathcal{E}}_0$  &0.045  & 0.036 & 0.055& 0.037&0.520 & 0.061& 0.060 \\


  \hline
    \end{tabular}
 \caption{Performance comparison between  TV-GL, ESA-GL, DTV-GL, DESA-GL, SpecTemp \cite{Segarrajournal2016}, Kalofolias \cite{Kalofolias} and Dong et al. \cite{Frossard_jour} for bandlimited, GS with discrete alphabet, over RG with   $K=2,6$ clusters.}\label{tab:comperison3}
\end{table*}

\section{Conclusions}
In this paper we have proposed  efficient strategies for recovering the graph topology from the observed data. The main idea is to associate a graph to the data in such a way that the observed signal looks like a band-limited signal over the inferred graph. The motivation is that enforcing the observed signal to be band-limited signal over the inferred graph tends to provide modular graphs and modularity is a structural property that carries important information.  The proposed method consists of two steps: learn first the sparsifying  transform and the graph signals jointly from the observed signals, and then infer the graph Laplacian from the estimated orthonormal bases.
Although the graph topology inference is intrinsically non-convex and we cannot make any claim of optimality on the achieved solution, the proposed algorithms are computationally efficient, because  the first step alternates between intermediate solutions expressed in closed form, while the second step involves a convex optimization.
We applied our methods to recover the brain functional activity network from ECoG seizure data, collected in an epilepsy study, and we have devised a purely data-driven strategy to assess the goodness of the inferred graph. Finally, we have compared our methods with existing ones under different settings and conditions to assess the relative merits.




\vspace{-0.28cm}
\bibliographystyle{IEEEbib}
\bibliography{sardel_refs_old}
\vspace{-0.35cm}

\appendices{}
\section{Closed-form solution for problem $\tilde{\mathcal{U}}_{k}$}
\label{A:closed_form_U} 
The proof is conceptually similar to that given in \cite{Manton},\cite{Bresler15},
for unitary transform, with the only difference that in our case one eigenvector is known a priori.
Hence,  given the sparse data matrix $\mS^{k} \in \mathbb{R}^{N \times M}$ and the observations matrix $\mY \in \mathbb{R}^{N \times M}$,   we derive closed form solution  for the optimization problem $\mathcal{U}_k$.
Note that the objective function is equivalent to
\beq \nonumber
\begin{split}
\parallel \mU^T \mY -{\mS}^{k}\parallel_{F}^{2} =\text{tr}\left( \mU^T \mY \mY^T \mU +\mS^{k} (\mS^{k})^T\!\!-
 2 \, \mU^T \mY (\mS^{k})^T\right)
\end{split}
\eeq
and using the orthonormality property $\mU^T \mU=\mI$,
problem $\mathcal{U}_k$  becomes
\beq
\begin{array}{lll}
\underset{\mathbf{U} \in \mathbb{R}^{N \times N}}{\max} \quad \text{tr}\left( \mU^T \mY (\mS^{k})^T\right) \quad \quad \quad (\mathcal{Q}_k)\\
\begin{array}{lll} \hspace{0.3cm}\mbox{s.t.}
      & \quad \, \mU^T \mU=\mI, \quad \bu_1=b \mathbf{1}\,.
    \end{array}
\end{array}
\eeq
Defining $\bar{\mY}^{k}= \mY (\mS^{k})^T$, it holds
\beq
\text{tr}\left( \mU^T \bar{\mY}^{k}\right)=\ds \sum_{i=1}^{N} \bu_i^T \bar{\my}_i^{k}=b \mathbf{1}^T \bar{\my}_1^{k}+ \ds \sum_{i=2}^{N} \bu_i^T \bar{\my}_i^{k}
\eeq
where $\bar{\my}_i^{k}$ represents the $i$th column of  $\bar{\mY}^{k}$.
Therefore, by introducing the matrices $ \bar{\mU}=[\bu_2,\ldots,\bu_N]$ and ${\mZ}^{k}=[\bar{\my}_2^{k},\ldots,\bar{\my}_N^{k}]$, problem $\mathcal{Q}_k$
is equivalent to the following non-convex problem
\beq
\begin{array}{lll}
\underset{\bar{\mathbf{U}} \in \mathbb{R}^{N \times (N-1)}}{\max} \quad \text{tr}\left( \bar{\mU}^T \mZ^{k} \right) \quad \quad \quad (\bar{\mathcal{Q}}_k)\\
\begin{array}{lll} \hspace{0.3cm} \mbox{s.t.}
      & \quad \quad \; \bar{\mU}^T \bar{\mU}=\mI_{N-1}, \quad \mathbf{1}^T\bar{\mU}=\mathbf{0}^T\,.
    \end{array}
\end{array} \label{prob_Z}
\eeq
The Lagrangian function associated to $\bar{\mathcal{Q}}_k$  can be written as
\beq
\begin{split} \nonumber
\mathcal{L}(\bar{\mU})=\text{tr}\left( \bar{\mU}^T \mZ^{k} \right) & -\text{tr}\left( \bar{\mLambda}(\bar{\mU}^T \bar{\mU}-\mI_{N-1})\right)  +\mathbf{1}^T\bar{\mU} \bar{\bmu}
\end{split}
\eeq
where $\bar{\mLambda} \in \mathbb{R}^{(N-1) \times (N-1)}$ and $\bar{\bmu} \in \mathbb{R}^{N-1 \times 1}$ contain the Lagrangian multipliers associated to the constraints.
Then, the KKT necessary conditions for the solutions optimality are
\beq
\begin{array}{lll}
\text{a)} \; \nabla_{\bar{\mU}}\mathcal{L}(\bar{\mU})=\mZ^{k}- \bar{\mU}(\bar{\mLambda}+\bar{\mLambda}^T) +  \mathbf{1}\bar{\bmu}^T=\mathbf{0}_{N \times N-1}\medskip\\
\text{b)} \; \bar{\mLambda} \; \bot \;(\bar{\mU}^T \bar{\mU}-\mI_{N-1})=\mathbf{0}\medskip\\
\text{c)}  \; \bar{\bmu} \; \bot \;\mathbf{1}^T\bar{\mU} =\mathbf{0}^T
\end{array}
\eeq
where $\mA  \bot  \mB$ stands for $\langle \mA , \mB \rangle \triangleq \text{tr}(\mA \mB^T)$.
From $\text{a)}$, assuming w.l.o.g. $\bar{\mLambda}=\bar{\mLambda}^T$, one gets
\beq \label{Qt}
2 \bar{\mU} \bar{\mLambda}=\mZ^{k}+\mathbf{1}\bar{\bmu}^T
\eeq
which, after multiplying both sides by $\mathbf{1}^T$ and using $\text{c)}$, gives  $\mathbf{1}^T(\mZ^{k}+\mathbf{1}\bar{\bmu}^T)=\mathbf{0}^T$
or
$ \bar{\bmu}^T=- \mathbf{1}^T\mZ^{k}/N$.
Plugging this last equality in (\ref{Qt}), we have
\beq \label{Qt1}
2 \bar{\mU}\bar{\mLambda}= \mP \mZ^{k}
\eeq
where $\mP=\mI-\mathbf{1} \mathbf{1}^T/N$.
Then, from  (\ref{Qt1}), it holds
\beq \label{Qt1n}
4 \bar{\mLambda}\bar{\mU}^T\bar{\mU} \mLambda=(\mZ^{k})^T \mP \mZ^{k} \Rightarrow \bar{\mLambda}=((\mZ^{k})^T \mP \mZ^{k})^{1/2}/2,
\eeq
where the last equality follows from $\text{b)}$.
Hence, replacing (\ref{Qt1n}) in (\ref{Qt1}), we get
\beq \label{Qt2}
\bar{\mU}((\mZ^{k})^T \mP \mZ^{k})^{1/2}=\mP \mZ^{k}.
\eeq
Observe that  $\mZ_{\perp}^{k}=\mP \mZ^{k}$ is the orthogonal projection of $\mZ^{k}$ onto the orthogonal
complement of the one-dimensional space
$\text{span}\{\mathbf{1}\}$,
so we yield
\beq \label{Qt3}
\bar{\mU}((\mZ^{k})^T_{\perp} \mZ_{\perp}^{k})^{1/2}=\mZ_{\perp}^{k}.
\eeq
Let denote with  ${\mZ}_{\perp}^{k}=\mX \mSigma \mV^T$ the  svd decomposition of ${\mZ}_{\perp}^{k}$ where $\mX \in \mathbb{R}^{N \times N}$, $\mV \in \mathbb{R}^{(N-1) \times (N-1)}$ and $\mSigma$ is a diagonal rectangular $N \times (N-1)$ matrix.
More specifically, if $r=\text{rank}({\mZ}_{\perp}^{k})$ we can rewrite ${\mZ}_{\perp}^{k}$ as
\beq \label{Zp}
{\mZ}_{\perp}^{k}=\mX \mSigma \mV^T=[\mX_r \, \mX_s]  \mSigma \mV^T
\eeq
where the $N\times r$ matrix $\mX_r$ contains as columns  the  $r$
left-eigenvectors associated to the non-zeros singular values of $\mZ_{\perp}^{k}$, $\mX_s$ is a  $N\times N-r$ matrix selected  as  belonging
to the nullspace of the matrix $\mB=[ \mathbf{1}^T; \mX_r^T]$, i.e. $\mB \mX_s=\mathbf{0}$.
This choice meets the orthogonality condition $\mathbf{1}^T \mX=\mathbf{0}^T$ with $\mX^T \mX=\mI_{N}$.
Therefore, by using ${\mZ}_{\perp}^{k}=\mX \mSigma \mV^T$  in  (\ref{Qt3}), we have
\beq
\label{Qt4}
\bar{\mU}(\mV \mSigma^T \mSigma \mV^T)^{1/2}=\mX \mSigma \mV^T
\eeq
or
$\bar{\mU}(\mV (\mSigma^T \mSigma)^{1/2} \mV^T)=\mX \mSigma \mV^T$, then
\beq
\bar{\mU} \mV (\mSigma^T \mSigma)^{1/2} =\mX \mSigma.
\eeq
Being $\mSigma$ a $N\times N-1$ rectangular diagonal matrix, it holds
\beq
\mX \mSigma=\mX^{-} \mSigma_{-}
\eeq
where $\mX^{-}$, $\mSigma_{-}$  are the matrices obtained by removing, respectively,  from $\mX$ its last column and from $\mSigma$ the last all zero row.
Hence, it holds
\beq
\label{Qt5}
\bar{\mU} \mV (\mSigma^T \mSigma)^{1/2} =\mX^{-} \mSigma_{-}\quad \Rightarrow \quad \bar{\mU} \mV \mSigma_{-}=\mX^{-} \mSigma_{-}
\eeq
then the optimal solution is
\beq
\bar{\mU}^{\star}=\mX^{-} \mV^T
\eeq
and
\beq
\mU^{\star}=[ \mathbf{1}b \quad \bar{\mU}^{\star}].
\eeq
Let us now  prove that $\bar{\mU}^{\star}$ is a global maximum for problem $\bar{\mathcal{Q}}_k$.
First observe that from the orthogonality condition $\bar{\mU}^T \mathbf{1}=\mathbf{0}$ we have
\beq \label{tr3}
\text{tr}\left( \bar{\mU}^T \mZ^{k} \right)=\text{tr}\left( \bar{\mU}^T \mZ_{\perp}^{k} \right)=\text{tr}\left( \mV^T \bar{\mU}^T \mX \mSigma \right).
\eeq
Defining  $\mQ=\mV^T \bar{\mU}^T \mX$,
it holds
\beq
\begin{split}
\mQ \mQ^T=& \mV^T \bar{\mU}^T \mX \mX^T \bar{\mU} \mV=\mI_{N-1}.
\end{split}
\eeq
From (\ref{tr3}) and using $\mQ^{\star}=\mV^T {\bar{\mU}^{\star \; T}} \mX= \mI_{N-1,N}$, with $\mI_{N-1,N}=[\mI_{N-1} \, \mathbf{0}]$, we have to prove that
\beq \label{ineq_opt}
\text{tr}\left( \mQ \mSigma \right)\leq \text{tr}\left( \mI_{N-1,N} \mSigma \right),\quad \forall \, \mQ \,: \,   \mQ \mQ^T=\mI_{N-1}
\eeq
with equality if and only if $\mQ^{\star}=\mI_{N-1,N}$.
The inequality  in (\ref{ineq_opt}) holds because $\Sigma_{ii}\geq 0$ and from $\mQ \mQ^T=\mI_{N-1}$ we yield $Q_{ii}\leq |Q_{ii}|\leq 1$ \cite{Manton}.
On the other hand $Q_{ii}=1$ $\forall i$ if and only if $\mQ=\mI_{N-1,N}$  so that the maximum is achieved in
$\mQ^{\star}$.

\vspace{-0.3cm}

\section{Derivation of the compact form in  (\ref{eq_matrix_system})}
\label{B:closed_form_X} 
In this section we show as equations a)-c) and f) in system (\ref{eq_system}) can be reduced to the matrix form in  (\ref{eq_matrix_system}).
 Note that, from conditions $\text{c)}$,  the vectorization decomposition of
  $\mL$ reads
 \beq
 \text{vec}(\mL)=\mM \, \text{vech}(\mL) \triangleq \mM \, \mz
 \eeq
 where $\text{vec}(\mL)$  is the $N^2$-dimensional column vector obtained by stacking the columns of the matrix $\mL$ on top of one another;
  $\mz \triangleq \text{vech}(\mL) \in \mathbb{R}^{N (N-1)/2}_{-}$ with $\text{vech}(\mL)$ the half-vectorization of $\mL$ obtained
  by  vectorizing only the lower triangular part (except the diagonal entries)  of $\mL$.  To define the  duplication matrix
   $\mM \in \mathbb{R}^{N^2 \times N(N-1)/2}$ which meets conditions $b)-f)$,  we first introduce   the  matrix $\overline{\mM}$, that  can be described in terms of its rows or columns as detailed next: for
  $i\geq j$,
   the $[(j-1)N+i]$th  and the $[(i-1)N+j]$th  rows of $\overline{\mM}$ equal the
   $[(j-1)(2 N-j)/2+i]$th
   row of $\mI_{N(N+1)/2}$.
       Then,   $\mM$  can be derived from $\overline{\mM}$ as follows: i) first remove from $\overline{\mM}$ the columns of index $(k-1)N+k-\sum_{l=1}^{k-1} l$ for $k=1,\ldots, N$ by defining the new matrix $\mM$; ii) replace in $\mM$ the rows of  index $k+N(k-1)$ with the vector $\mv$  derived by summing up the rows of the
      matrix $\mD_k=-\mM(1+N (k-1):N k,:)$ for $k=1,\ldots,N$.
   As an example, if we set $N=3$, the matrix $\mM$ reads as
      \beq
      \mM=\left(
            \begin{array}{rrrrrrrrr}
            -1 & 1 & 0 &1 & -1 & 0 & 0 & 0 &0 \\
            -1 & 0 & 1 & 0 & 0 & 0 & 1 & 0 & -1\\
            0 & 0 & 0 & 0 & -1 & 1 & 0 & 1 & -1
            \end{array}
          \right)^T.
      \eeq
Therefore, the first equation in (\ref{eq_system}) can be vectorized as $\text{vec}(\mL \mU_{\mathcal{K}})-\text{vec}(\mU_{\mathcal{K}} \mLambda_{\mathcal{K}})=\mathbf{0}$. We can now exploit the property of the vec operator, namely $\text{vec}(\mA\mB)=(\mI \otimes \mA)\text{vec}(\mB)=(\mB^T \otimes \mI)\text{vec}(\mA)$ and define to that purpose the matrices $\mB\in \mathbb{R}^{K N \times N (N-1)/2}$ and $\mQ \in \mathbb{R}^{K N \times K}$ as
\beq \nonumber
\mB=(\mU_{\mathcal{K}}^T \otimes \mI_N) \mM,\quad \quad \mQ=(\mI_K \otimes \mU_{\mathcal{K}}) \left( \sum_{k=1}^K \me_k \otimes \mE_k\right)
\eeq
where $\me_k$ denotes the canonical $K$-dimensional vector with the $k$-entry equal to $1$
and $\mE_k= \me_k \cdot \me_k^T$. Note that the matrix $\mQ$ is full-column rank (with rank $K$)
since each column $k$ is $[\mathbf{0}; \ldots;\underset{k-th}{\underbrace{\mathbf{u}_k}};\mathbf{0}]$.
Let us introduce the matrix $\mQ_{-}$ obtained  by removing from $\mQ$ the column with index corresponding to $\lambda_{\mathcal{K},1}$.
 Hence, one can easily see that part of  system (\ref{eq_system}) reduces to the following form
\beq \label{eq_matrix_system1}
 \begin{array}{lllll} \mF \mx=\mb, \quad \quad \mx \in \mathbb{R}^{ N (N-1)/2+ K-1}_{+}
\end{array}
\eeq
where we defined  the coefficient matrix $\mF \in \mathbb{R}^{m\times n}$,
with $m=K N+1$, $n=N (N-1)/2+ K-1$,  as
\beq \label{eq:F_matrix}
\mF \triangleq \left[\begin{array}{ll} -\mB & -\mQ_{-}\\ \mathbf{1}_{ N(N-1)/2}^T & \mathbf{0}^T_{K-1} \end{array}\right],
\eeq
$\mx\triangleq [-\mz ; \bar{\boldmath{\mlambda}}]$; $\bar{\boldmath{\mlambda}}\triangleq \{\lambda_{\mathcal{K},i}\}_{\forall i \in \mathcal{K}^{-}}$
where, assuming the entries of $\boldmath{\mlambda}_{\mathcal{K}}$ in increasing order,  the index set $\mathcal{K}^{-}$ is obtained by removing from $\mathcal{K}$ the first index corresponding to $\lambda_{\mathcal{K},1}$; and, finally, $\mb=[\mathbf{0}_{K N};\, p]$.

\vspace{-0.2cm}
 \section{}

\subsection{Proof of Proposition $2$}
Assume that the set is feasible, i.e. there exists at least a point $\mx\in \mathcal{X}(\mU_{\mathcal{K}})$.
The solution set size depends on the rank $q$ of the coefficient matrix $\mF$ and of the augmented matrix $[ \mF, \, \mb]$.  Let us distinguish the two cases $m\leq n$, which leads to $K \leq \frac{N}{2}-\frac{2}{N-1}$, and $m>n$, or $K > \frac{N}{2}-\frac{2}{N-1}$.
Then,  necessary conditions for which the system  (\ref{eq_matrix_system}) admits at least a solution are:
$q=\text{rank}(\mF)=\text{rank}([\mF, \, \mb])\leq m=K N+1$ for $K \leq \frac{N}{2}-\frac{2}{N-1}$; $q=\text{rank}(\mF)=\text{rank}([\mF, \, \mb])\leq n$ for $K > \frac{N}{2}-\frac{2}{N-1}$. Additionally, since the rank of $\mQ_{-}$ is  $K-1$, it holds $K-1 \leq \text{rank}(\mF)$. This proves statement a).
To show b) note that, from the feasibility of $\mathcal{X}(\mU_\mathcal{K})$, the condition $q=m=n$ or $K=\frac{N}{2}-\frac{2}{N-1}$ is  sufficient in order for the set to be   a singleton. Additionally, if $K >\frac{N}{2}-\frac{2}{N-1}$ and $\text{rank}(\mF)=n=N (N-1)/2+ K-1$,
 the set is a singleton.

\subsection{Proof of Proposition $3$}
Assume  the set $\mathcal{X}(\mU_\mathcal{K})$ with $m<n$  feasible and  $\{\mx \, | \, \mF \mx=\mF \mx_0,  \mx \geq \mathbf{0}\}$  a singleton  for any nonnegative $s$-sparse signal $\mx_0$.
This implies $\{\mx \, | \, \mF \mx=\mF \mx_0, \mx \geq \mathbf{0}\}=\mx_0$ with $\mx_0$ a non negative $s$-sparse vector.
Then, from Proposition $1$ in \cite{Wang} if $\mF \in \mathbb{R}^{m\times n}$ with $m<n$,  or $\frac{N}{2}-\frac{2}{N-1} > K$,
it holds  that $m\geq 2 s+1$, i.e. $K\geq 2 s/N$.  Therefore we get  $\frac{N}{2}-\frac{2}{N-1} > K\geq 2 s/N$.  Additionally, from the feasibility of $\mathcal{X}(\mU_{\mathcal{K}})$, if
   $c$ denotes the multiplicity of the eigenvalue $0$ of  $\mL$, i.e. the number of connected components of the graph,
 it results  $\parallel \!\bar{\boldmath{\mlambda}} \parallel_0=K-c$,
so that  $s=K-c+\parallel \mA \parallel_0/2$. This implies from $m\geq 2 s+1$ that $K N\geq 2 K-2 c+ \parallel \mA \parallel_0$ or
$\parallel \mA \parallel_0\leq K(N-2)+2 c$.

\end{document}